\documentclass[11pt]{article}
\usepackage[letterpaper,hmargin=1in,vmargin=1.25in]{geometry}

\usepackage{url}

\usepackage{amsmath}
\usepackage{amssymb}

\usepackage{amsthm}

\theoremstyle{plain}
\newtheorem{theorem}{Theorem}[section]

\newtheorem{corollary}[theorem]{Corollary}

\newtheorem{thesis}{Thesis}

\newtheorem{assumption}{Assumption}

\theoremstyle{definition}
\newtheorem{definition}[theorem]{Definition}

\newtheorem{protocol}[theorem]{Protocol}
\newtheorem{postulate}{Postulate}
\newtheorem{remark}[theorem]{Remark}

\renewcommand{\labelenumi}{(\roman{enumi})}

\newcommand{\abs}[1]{\left\lvert#1\right\rvert}

\newcommand{\rest}[2]{#1\!\!\restriction_{#2}}

\newcommand{\osg}[1]{\left[#1\right]^{\prec}}

\newcommand{\N}{\mathbb{N}}%

\newcommand{\Bm}[2]{\lambda_{#1}\left(#2\right)}

\newcommand{\ket}[1]{| #1 \rangle}
\newcommand{\bra}[1]{\langle #1 |}
\newcommand{\braket}[2]{\langle #1 | #2 \rangle}
\newcommand{\product}[2]{\lvert #1\rangle\langle #2\rvert}

\newcommand{\PS}{\mathbb{P}}%

\newcommand{\cond}[2]{\mathrm{Filtered}_{#1}\left(#2\right)}

\newcommand{\ssoa}{\overline{\mathcal{H}}}

\newcommand{\noi}{\noindent}

\title{\textbf{A refinement of the argument of
local realism
versus quantum mechanics
by algorithmic randomness}\thanks{%
This work
is a substantial extension of a preliminary paper of the author,
entitled:
``A refinement of the argument of Bell's inequality versus quantum mechanics by algorithmic randomness,''
which appeared
without peer review
in: Masahiko Sakai (ed.), New Trends in Algorithms and Theory of Computation,
RIMS K\^{o}ky\^{u}roku, No.2154, pp.27--53, April 2020.
The preliminary paper is available at
\texttt{http://hdl.handle.net/2433/255108}%
}}

\author{Kohtaro Tadaki\\
\\
Department of Computer Science, College of Engineering, Chubu University\\
1200 Matsumoto-cho, Kasugai-shi, Aichi 487-8501, Japan\\
E-mail: \textsf{tadaki@fsc.chubu.ac.jp}\\
\url{https://tadaki.org/}}

\date{
\begin{quotation}
\noi\textbf{Abstract.}
The notion of probability plays a crucial role in quantum mechanics.
It appears in quantum mechanics as the \emph{Born rule}.
In modern mathematics which describes quantum mechanics, however, probability theory means nothing other than measure theory,
and therefore any operational characterization of the notion of probability is still missing in quantum mechanics.
In our former works~[K.~Tadaki, arXiv:1804.10174],
based on the toolkit of \emph{algorithmic randomness},
we presented
a refinement of the Born rule,
called the \emph{principle of typicality},
for specifying the property of results of measurements \emph{in an operational way}.
In this paper, we make an application of our framework to
the argument of local realism versus quantum mechanics
for refining it,
in order to demonstrate
\emph{how properly our framework works
in practical problems in quantum mechanics}.
\end{quotation}
\begin{quotation}
\noi\textit{Key words\/}:
the principle of typicality,
Born rule,
probability interpretation,
algorithmic randomness,
operational characterization,
Martin-L\"of randomness,
many-worlds interpretation,
local realism,
Bell's inequality,
GHZ experiment
\end{quotation}
}

\begin{document}

\maketitle

\section{Introduction}

The notion of probability plays a crucial role in quantum mechanics.
It appears in quantum mechanics as the so-called \emph{Born rule}, i.e.,
\emph{the probability interpretation of the wave function} \cite{D58,vN55,NC00}.
In modern mathematics which describes quantum mechanics, however,
probability theory means nothing other than \emph{measure theory},
and therefore any \emph{operational characterization of the notion of probability} is still missing
in quantum mechanics.
In this sense, the current form of quantum mechanics is considered to be \emph{imperfect}
as a physical theory which must stand on operational means.

In the work~\cite{T14,T15,T15CCR,T16arXiv,T19arXiv}
we developed a theory of
\emph{operational characterization of the notion of probability by algorithmic randomness}
for general discrete probability spaces,
using the notion of \emph{Martin-L\"of randomness with respect to Bernoulli measure} \cite{M66} (and its extension over the Baire space).
In the work
we gave natural and equivalent operational characterizations of the basic notions of probability theory,
such as the notions of conditional probability and the independence of events/random variables,
in terms of
the notion of Martin-L\"of randomness with respect to Bernoulli measure.
We then made
applications of this
framework
to information theory and cryptography, as examples of the applications,
in order to demonstrate the wide applicability of
our
framework
to the general areas of
science and
technology.
The crucial point is that
in our theory
the underlying discrete probability space is \emph{quite arbitrary}, and
thus
we do not
need to
impose any computability restrictions on the underlying
discrete
probability space
at all,
in particular.
See Tadaki~\cite{T16arXiv}
for the details of our theory of operational characterization of the notion of probability
in the case of
general finite probability spaces
and its applications,
and Tadaki~\cite{T19arXiv}
for ones
in the case of general discrete probability spaces
whose sample spaces are countably infinite.

As a major application of our framework~\cite{T14,T15,T15CCR,T16arXiv} to basic science,
in the work~\cite{T14CCR,T15QCOMPINFO,T15Kokyuroku,T16CCR,T16QIP,T16QIT35,T17SCIS,T18arXiv}
we presented an \emph{operational refinement} of the Born rule, as an alternative rule to it,
based on our theory of operational characterization of the notion of probability
by algorithmic randomness,
for the purpose of making quantum mechanics \emph{operationally perfect}.
Namely, we used the notion of
\emph{Martin-L\"of randomness with respect to Bernoulli measure}
to present
the operational refinement of the Born rule,
for specifying the property of the results of
quantum measurements \emph{in an operational way}.
We then presented
an operational refinement of the Born rule for mixed states, as an alternative rule to it,
based on Martin-L\"of randomness with respect to Bernoulli measure.
In particular, we gave a
\emph{precise definition}
for the notion of \emph{mixed state}.
We then showed that all of
the refined rules of the Born rule for both pure states and mixed states
can be derived
from a \emph{single} postulate, called the \emph{principle of typicality}, in a unified manner.
We did this from the point of view of the \emph{many-worlds interpretation of quantum mechanics}~\cite{E57}.
Finally, we made an application of our framework to the BB84 quantum key distribution protocol
in order to demonstrate how properly our framework works in
\emph{practical}
problems
in quantum mechanics, based on the principle of typicality.
See the work~\cite{T18arXiv} for the details of our framework of quantum mechanics based on the principle of typicality.

Around 2000, there were comprehensive attempts
to extend the algorithmic randomness notions
to the quantum regime.
Namely, Berthiaume, van Dam, and Laplante~\cite{BvDL00},
Vit\'{a}nyi~\cite{Vit00},
and G\'{a}cs~\cite{Gac01} independently introduced
the so-called \emph{quantum Kolmogorov complexity}
from each standpoint,
for aiming to define the algorithmic information content of a given individual quantum state of a finite dimensional quantum system.
Later on, Tadaki~\cite{T06MLQ} introduced in 2004
the notion of $\hat{\Omega}$
by means of modifying the work of G\'{a}cs~\cite{Gac01}.
This $\hat{\Omega}$
is an extension of Chaitin's halting probability $\Omega$~\cite{C75} to
a measurement operator in an infinite dimensional quantum system.
Actually, $\hat{\Omega}$ is a bounded Hermitian operator on a Hilbert space of infinite dimension,
and for every computable state $\ket{\Psi}$ it holds that
the inner-product $\bra{\Psi}\hat{\Omega}\ket{\Psi}$,
which has a meaning of ``probability'' in quantum mechanics,
is a Martin-L\"of random recursively enumerable real
and therefore can be represented as a Chaitin's $\Omega$.
These previous works are about the extension of algorithmic randomness notions
to the quantum realm.

In contrast, in this paper as well as our former works~\cite{T14CCR,T15QCOMPINFO,T15Kokyuroku,T16CCR,T16QIP,T16QIT35,T17SCIS,T18arXiv},
we do not extend the algorithmic randomness notions themselves,
while keeping them in their original forms,
but we identify in the mathematical framework of quantum mechanics
the basic idea of Martin-L\"of randomness
that \emph{the random infinite sequences are precisely sequences which are not contained in any effective null set}.
See Martin-L\"of~\cite{M66}, Nies~\cite{N09}, Downey and Hirschfeldt~\cite{DH10},
and Brattka, Miller, and Nies~\cite{BMiN12} for this basic idea of Martin-L\"of randomness.
In our works, we do this identification, in particular,
from the aspect of the many-worlds interpretation of quantum mechanics~\cite{E57},
leading to the principle of typicality.

\subsection{Contributions of the paper}

In this paper, we make an application of our framework
based on the principle of typicality
to the \emph{argument of local realism versus quantum mechanics}
to refine it operationally,
in order to
\emph{further}
demonstrate how properly our framework works
in \emph{practical} problems in quantum mechanics.

Specifically,
in this paper,
first
we \emph{refine} and \emph{reformulate}
the argument of \emph{Bell's inequality versus quantum mechanics} initiated by Bell~\cite{Bell64},
by means of algorithmic randomness.
In particular, we refine and reformulate a version of the argument
exposited
in
Section~2.6 ``EPR and the Bell inequality'' of Nielsen and Chuang~\cite{NC00},
where the \emph{assumptions of local realism}
result in
the so-called \emph{CHSH inequality}~\cite{CHSH69}.
Thus,
on the one hand,
we refine and reformulate the assumptions of local realism themselves
and the derivation from them to
the CHSH inequality
\[
  \langle RS\rangle + \langle QS\rangle + \langle RT\rangle - \langle QT\rangle\le 2,
\]
i.e., the inequality~\eqref{Bell-inequality} below,
in terms of
the theory of
operational characterization of the notion of probability by algorithmic randomness,
developed by the work~\cite{T14,T15,T15CCR,T16arXiv} as mentioned above.
On the other hand,
we refine and reformulate
the corresponding argument of quantum mechanics
to violate Bell's inequality, resulting in
the equality
\[
  \langle RS\rangle + \langle QS\rangle + \langle RT\rangle - \langle QT\rangle=2\sqrt{2},
\]
i.e., the equality~\eqref{QM-equality} below,
based on the principle of typicality.

Second, in this paper we \emph{refine} and \emph{reformulate}
an entirely different type of the variant of
the argument of Bell's inequality versus quantum mechanics, i.e.,
the argument over the \emph{GHZ experiment} \cite{GHZ89,GHSZ90,Mer90},
by means of algorithmic randomness.
In particular, we refine and reformulate
the analysis for
the GHZ experiment
by quantum mechanics,
proposed by Mermin~\cite{Mer90}
for a system of three spin-$1/2$ particles,
in terms of algorithmic randomness.
Note that
this Mermin's work is, in turn, based on the pioneering work of
Greenberger, Horne, and Zeilinger~\cite{GHZ89},
who
first proposed ``Bell's Theorem Without Inequalities'' via the GHZ experiment.
Thus,
on the one hand,
we refine and reformulate
the argument of quantum mechanics which leads to
\emph{perfect correlations of measurement results over three parties}
for a system of three spin-$1/2$ particles,
i.e., Theorem~\ref{thm:GHZ} below, based on the principle of typicality.
On the other hand,
we refine and reformulate the argument that
the analysis of the GHZ experiment based on the assumptions of local realism
cannot recover the prediction of quantum mechanics, i.e.,
the perfect correlations over three parties,
in terms of
our theory of operational characterization of the notion of probability
by algorithmic randomness.

\subsection{Organization of the paper}

The paper is organized as follows.
We begin in Section~\ref{preliminaries} with
some mathematical preliminaries, in particular, about measure theory and
Martin-L\"of randomness with respect to an arbitrary probability measure.
We then review the notion of Martin-L\"of randomness with respect to
an arbitrary
Bernoulli measure,
called the \emph{Martin-L\"of $P$-randomness} in this paper,
in Section~\ref{Sec-ML_P-randomness}.
In Section~\ref{FMP} we summarize
the theorems and notions on Martin-L\"of $P$-randomness from Tadaki~\cite{T14,T15,T16arXiv},
which are need to establish the contributions of this paper presented
in Sections~\ref{sec-QM-equality}, \ref{sec-Bell_inequality},
\ref{sec:GHZ-QM}, and \ref{sec:GHZ-imperfect}.
In Section~\ref{QM} we review the central postulates of the \emph{conventional} quantum mechanics
according to Nielsen and Chuang~\cite{NC00}.
In Section~\ref{MWI}
we review the framework of the \emph{principle of typicality},
which was introduced by Tadaki~\cite{T16CCR,T16QIP,T16QIT35,T17SCIS,T18arXiv}.

In Sections~\ref{sec-QM-equality} and \ref{sec-Bell_inequality},
we refine and reformulate the argument of local realism
versus quantum mechanics,
in the context of Bell's inequality~\cite{Bell64,CHSH69,NC00}.
On the one hand,
in Section~\ref{sec-QM-equality},
we refine and reformulate the argument of quantum mechanics to violate Bell's inequality,
based on the principle of typicality.
On the other hand,
in Section~\ref{sec-Bell_inequality},
we refine and reformulate the assumptions of local realism themselves and the derivation from them to
Bell's inequality,
in terms of
the theory of
operational characterization of the notion of probability
by algorithmic randomness \cite{T14,T15,T15CCR,T16arXiv}.

In Sections~\ref{sec:GHZ-QM} and \ref{sec:GHZ-imperfect},
we refine and reformulate the argument of local realism
versus quantum mechanics,
in the context of the GHZ experiment~\cite{GHZ89,GHSZ90,Mer90}.
On the one hand,
in Section~\ref{sec:GHZ-QM},
we refine and reformulate
the analysis of the GHZ experiment by quantum mechanics,
leading to the perfect correlations
over three parties,
based on the principle of typicality.
On the other hand,
in Section~\ref{sec:GHZ-imperfect},
we refine and reformulate
the demonstration of the underivability of the perfect correlation
from local realism,
in terms of
our theory of
operational characterization of the notion of probability
by algorithmic randomness.
We conclude this paper with
summary
in Section~\ref{sec-Concluding_remarks}.

\section{Mathematical preliminaries}
\label{preliminaries}

\subsection{Basic notation and definitions}
\label{basic notation}

We start with some notation about numbers and strings which will be used in this paper.
We denote the \emph{cardinality} of $S$ by $\#S$ for any set $S$.
$\N=\left\{0,1,2,3,\dotsc\right\}$ is the set of \emph{natural numbers},
and $\N^+$ is the set of \emph{positive integers}.
For any complex number $z$, its \emph{complex conjugate} is denoted by
$\overline{z}$.

An \emph{alphabet} is a non-empty finite set.
Let $\Omega$ be an arbitrary alphabet throughout the rest of this subsection.
A \emph{finite string over $\Omega$} is a finite sequence of elements from the alphabet $\Omega$.
We use $\Omega^*$ to denote the set of all finite strings over $\Omega$,
which contains the \emph{empty string} denoted by $\lambda$.
For any $\sigma\in\Omega^*$, $\abs{\sigma}$ is the \emph{length} of $\sigma$.
Therefore $\abs{\lambda}=0$.
A subset $S$ of $\Omega^*$ is called
\emph{prefix-free}
if no string in $S$ is a prefix of another string in $S$.

An \emph{infinite sequence over $\Omega$} is an infinite sequence of elements from the alphabet $\Omega$,
where the sequence is infinite to the right but finite to the left.
We use $\Omega^\infty$ to denote the set of all infinite sequences over $\Omega$.

Let $\alpha\in\Omega^\infty$.
For any $n\in\N$
we denote by $\rest{\alpha}{n}\in\Omega^*$ the first $n$ elements
in the infinite sequence $\alpha$,
and
for any $n\in\N^+$ we denote
by $\alpha(n)$ the $n$th element in $\alpha$.
Thus, for example, $\rest{\alpha}{4}=\alpha(1)\alpha(2)\alpha(3)\alpha(4)$, and $\rest{\alpha}{0}=\lambda$.

For any $S\subset\Omega^*$, the set
$\{\alpha\in\Omega^\infty\mid\exists\,n\in\N\;\rest{\alpha}{n}\in S\}$
is denoted by $\osg{S}$.
Note that (i)~$\osg{S}\subset\osg{T}$ for every $S\subset T\subset\Omega^*$, and
(ii)~for every set $S\subset\Omega^*$ there exists a prefix-free set $P\subset\Omega^*$ such that
$\osg{S}=\osg{P}$.
For any $\sigma\in\Omega^*$, we denote by $\osg{\sigma}$ the set $\osg{\{\sigma\}}$, i.e.,
the set of all infinite sequences over $\Omega$ extending $\sigma$.
Therefore $\osg{\lambda}=\Omega^\infty$.

Let $\mathcal{H}$ be an arbitrary complex Hilbert space.
A \emph{projector} $P$ on $\mathcal{H}$ is
a non-zero Hermitian bounded operator on $\mathcal{H}$ such that $P^2=P$.
A collection $\{P_a\}_{a\in\Theta}$ of projectors on $\mathcal{H}$ is called
a \emph{projection-valued measure} (\emph{PVM}, for short) \emph{in $\mathcal{H}$} if
$\Theta$ is an alphabet and the following~(i) and (ii) hold:
\begin{enumerate}
\item $P_a P_b=\delta_{a,b}P_a$ for every $a,b\in\Theta$.
\item $\sum_{a\in\Theta} P_a=I$, where $I$ denotes the identity operator on $\mathcal{H}$.
\end{enumerate}
Thus, in the present paper, we only handle a PVM in a complex Hilbert space with a \emph{finite} number of PVM elements.
It is sufficient to deal with such a PVM for developing our general theory.

\subsection{Martin-L\"of randomness with respect to an arbitrary probability measure}
\label{MLRam}

We briefly review measure theory according to Nies~\cite[Section 1.9]{N09}.
See also Billingsley~\cite{B95} for measure theory in general.

Let $\Omega$ be an arbitrary alphabet.
A real-valued function $\mu$ defined on the class of all subsets of $\Omega^\infty$ is called
an \emph{outer measure on $\Omega^\infty$} if the following conditions hold:
\begin{enumerate}
  \item $\mu\left(\emptyset\right)=0$;
  \item $\mu\left(\mathcal{C}\right)\le\mu\left(\mathcal{D}\right)$
    for every subsets $\mathcal{C}$ and $\mathcal{D}$ of $\Omega^\infty$
    with $\mathcal{C}\subset\mathcal{D}$;
  \item $\mu\left(\bigcup_{i}\mathcal{C}_i\right)\le\sum_{i}\mu\left(\mathcal{C}_i\right)$
    for every sequence $\{\mathcal{C}_i\}_{i\in\N}$ of subsets of $\Omega^\infty$.
\end{enumerate}
A \emph{probability measure representation over $\Omega$} is
a function $r\colon\Omega^*\to[0,1]$ such that
\begin{enumerate}
  \item $r(\lambda)=1$ and
  \item for every $\sigma\in\Omega^*$ it holds that
    \begin{equation}\label{pmr}
       r(\sigma)=\sum_{a\in\Omega}r(\sigma a).
    \end{equation}
\end{enumerate}
A probability measure representation $r$ over $\Omega$ \emph{induces}
an outer measure $\mu_r$ on $\Omega^\infty$ in the following manner:
A subset $\mathcal{R}$ of $\Omega^\infty$ is called \emph{open} if
$\mathcal{R}=\osg{S}$ for some $S\subset\Omega^*$.
Let $r$ be an arbitrary probability measure representation over $\Omega$. 
For each open subset $\mathcal{A}$ of $\Omega^\infty$, we define $\mu_r(\mathcal{A})$ by
$$\mu_r(\mathcal{A}):=\sum_{\sigma\in E}r(\sigma),$$
where $E$ is a prefix-free subset of $\Omega^*$ with $\osg{E}=\mathcal{A}$.
Due to the equality~\eqref{pmr} the sum is independent of the choice of  the prefix-free set $E$,
and therefore the value $\mu_r(\mathcal{A})$ is well-defined.
Then, for any subset $\mathcal{C}$ of $\Omega^\infty$, we define $\mu_r(\mathcal{C})$ by
$$\mu_r(\mathcal{C}):=
\inf\{\mu_r(\mathcal{A})\mid
\mathcal{C}\subset\mathcal{A}\text{ \& $\mathcal{A}$ is an open subset of $\Omega^\infty$}\}.$$
We can then show that $\mu_r$ is an \emph{outer measure} on $\Omega^\infty$  such that
$\mu_r(\Omega^\infty)=1$.

A class $\mathcal{F}$ of subsets of $\Omega^\infty$ is called
a \emph{$\sigma$-field on $\Omega^\infty$}
if  $\mathcal{F}$ includes $\Omega^\infty$, is closed under complements,
and is closed under the formation of countable unions.
The \emph{Borel class} $\mathcal{B}_{\Omega}$ is the $\sigma$-field \emph{generated by}
all open sets on $\Omega^\infty$.
Namely, the Borel class $\mathcal{B}_{\Omega}$ is defined
as the intersection of all the $\sigma$-fields on $\Omega^\infty$ containing
all open sets on $\Omega^\infty$.
A real-valued function $\mu$ defined on the Borel class $\mathcal{B}_{\Omega}$ is called
a \emph{probability measure on $\Omega^\infty$} if the following conditions hold:
\begin{enumerate}
  \item $\mu\left(\emptyset\right)=0$ and $\mu\left(\Omega^\infty\right)=1$;
  \item $\mu\left(\bigcup_{i}\mathcal{D}_i\right)=\sum_{i}\mu\left(\mathcal{D}_i\right)$
    for every sequence $\{\mathcal{D}_i\}_{i\in\N}$ of sets in $\mathcal{B}_{\Omega}$ such that
    $\mathcal{D}_i\cap\mathcal{D}_i=\emptyset$ for all $i\neq j$.
\end{enumerate}
Then, for every probability measure representation $r$ over $\Omega$,
we can show that the restriction of the outer measure $\mu_r$ on $\Omega^\infty$
to the Borel class $\mathcal{B}_{\Omega}$ is
a probability measure on $\Omega^\infty$.
We denote the restriction of $\mu_r$ to $\mathcal{B}_{\Omega}$ by
$\mu_r$
just the same.

Then it is easy to see that
\begin{equation}\label{mr}
  \mu_r\left(\osg{\sigma}\right)=r(\sigma)
\end{equation}
for every probability measure representation $r$ over $\Omega$ and every $\sigma\in \Omega^*$.
The probability measure $\mu_r$ is called a
\emph{probability measure induced by the probability measure representation $r$}.

Now,
we introduce the notion of \emph{Martin-L\"of randomness} \cite{M66}
in a general setting,
as follows.

\begin{definition}[Martin-L\"{o}f randomness with respect to a probability measure]
\label{ML-randomness-wrtm}
Let $\Omega$ be an alphabet, and let $\mu$ be a probability measure on $\Omega^\infty$.
A subset $\mathcal{C}$ of $\N^+\times\Omega^*$ is called a
\emph{Martin-L\"{o}f test with respect to $\mu$} if
$\mathcal{C}$ is a recursively enumerable set and for every $n\in\N^+$ it holds that
$\mathcal{C}_n$ is a prefix-free subset of $\Omega^*$ and
\begin{equation}\label{muocn<2n}
  \mu\left(\osg{\mathcal{C}_n}\right)<2^{-n},
\end{equation}
where $\mathcal{C}_n$ denotes the set
$\left\{\,
    \sigma\mid (n,\sigma)\in\mathcal{C}
\,\right\}$.

For any $\alpha\in\Omega^\infty$, we say that $\alpha$ is
\emph{Martin-L\"{o}f random with respect to $\mu$} if
$$\alpha\notin\bigcap_{n=1}^{\infty}\osg{\mathcal{C}_n}$$
for every Martin-L\"{o}f test $\mathcal{C}$ with respect to $\mu$.\qed
\end{definition}

\section{\boldmath Martin-L\"of $P$-randomness}
\label{Sec-ML_P-randomness}

The principle of typicality, Postulate~\ref{POT} below, is stated by means of the notion of 
Martin-L\"of randomness with respect to an arbitrary probability measure
introduced in the preceding section.
However,
in many situations of the applications of the principle of typicality,
such as in a contribution of this paper described
in Sections~\ref{sec-QM-equality} and \ref{sec:GHZ-QM},
a more restricted notion is used where
the probability measure is chosen to be
a
Bernoulli measure.
Specifically,
the notion of \emph{Martin-L\"of randomness with respect to
a
Bernoulli measure}
is used
in many situations of the applications
of the principle of typicality.
Thus,
in order to introduce this notion, we first
review the notions of \emph{finite probability space} and \emph{Bernoulli measure}.

\begin{definition}[Finite probability space]\label{def-FPS}
Let $\Omega$ be an alphabet. A \emph{finite probability space on $\Omega$} is a function $P\colon\Omega\to [0,1]$
such that
\begin{enumerate}
  \item $P(a)\ge 0$ for every $a\in \Omega$, and
  \item $\sum_{a\in \Omega}P(a)=1$.
\end{enumerate}
The set of all finite probability spaces on $\Omega$ is denoted by $\PS(\Omega)$.

Let $P\in\PS(\Omega)$.
The set $\Omega$ is called the \emph{sample space} of $P$,
and elements
of
$\Omega$ are called \emph{sample points} or \emph{elementary events}
of $P$.
For each $A\subset\Omega$, we define $P(A)$ by
$$P(A):=\sum_{a\in A}P(a).$$
A subset of $\Omega$ is called an \emph{event} on $P$, and
$P(A)$ is called the \emph{probability} of $A$
for every event $A$
on $P$.
\qed
\end{definition}

Let $\Omega$ be an alphabet, and let $P\in\PS(\Omega)$.
For each $\sigma\in\Omega^*$, we use $P(\sigma)$ to denote
\[
  P(\sigma_1)P(\sigma_2)\dots P(\sigma_n)
\]
where $\sigma=\sigma_1\sigma_2\dots\sigma_n$ with $\sigma_i\in\Omega$.
Therefore $P(\lambda)=1$, in particular.
For each subset $S$ of $\Omega^*$, we use $P(S)$ to denote
$$\sum_{\sigma\in S}P(\sigma).$$
Therefore $P(\emptyset)=0$, in particular.

Consider a function $r\colon\Omega^*\to[0,1]$ such that
$r(\sigma)=P(\sigma)$ for every $\sigma\in\Omega^*$.
It is then easy to see that the function~$r$ is a probability measure representation over $\Omega$.
The probability measure~$\mu_r$ induced by
the probability measure representation
$r$ is
called
the \emph{Bernoulli measure on $\Omega^\infty$} (\emph{induced by the finite probability space $P$}), denoted $\lambda_{P}$.
The Bernoulli measure~$\lambda_{P}$ on $\Omega^\infty$
satisfies that
\begin{equation}\label{BmPosgsigma=Psigma}
  \Bm{P}{\osg{\sigma}}=P(\sigma)
\end{equation}
for every $\sigma\in \Omega^*$,
which follows from \eqref{mr}.

The notion of \emph{Martin-L\"of randomness with respect to a Bernoulli measure} is defined as follows.
We call it the \emph{Martin-L\"of $P$-randomness},
since it depends on a finite probability space~$P$.
This notion was, in essence, introduced by Martin-L\"{o}f~\cite{M66},
as well as the notion of Martin-L\"of randomness
with respect to Lebesgue measure.

\begin{definition}[%
Martin-L\"of $P$-randomness,
Martin-L\"{o}f \cite{M66}]\label{ML_P-randomness}
Let $\Omega$ be an alphabet, and let $P\in\PS(\Omega)$.
For any $\alpha\in\Omega^\infty$, we say that $\alpha$ is \emph{Martin-L\"{o}f $P$-random} if
$\alpha$ is Martin-L\"{o}f random with respect to $\lambda_{P}$.\qed
\end{definition}

\emph{Martin-L\"{o}f random} infinite binary sequences \cite{M66} are precisely 
Martin-L\"{o}f $U$-random infinite binary sequences
where $U$ is a finite probability space on $\{0,1\}$ such that $U(0)=U(1)=1/2$.

\section{\boldmath The properties of Martin-L\"of $P$-randomness}
\label{FMP}

In order to
obtain the results in this paper,
we need
several
theorems
and notions
on Martin-L\"of $P$-randomness
from
Tadaki~\cite{T14,T15,T16arXiv}.
Originally,
these theorems
and notions
played
a key part
in developing
the theory of
\emph{operational characterization of the notion of probability}
based on Martin-L\"of $P$-randomness in
Tadaki~\cite{T14,T15,T16arXiv}.
We enumerate
these theorems and notions
in this section.

\subsection{The law of large numbers}

The following theorem shows that the law of large numbers holds
for an arbitrary Martin-L\"of $P$-randomness infinite sequence.
For the proof of Theorem~\ref{FI},
see Tadaki~\cite[Theorem~14]{T16arXiv}.

\begin{theorem}[The law of large numbers, Tadaki~\cite{T14}]\label{FI}
Let $\Omega$ be an alphabet, and let $P\in\PS(\Omega)$.
For every $\alpha\in\Omega^\infty$, if $\alpha$ is Martin-L\"of $P$-random
then for every $a\in\Omega$ it holds that
$$\lim_{n\to\infty} \frac{N_a(\rest{\alpha}{n})}{n}=P(a),$$
where $N_a(\sigma)$ denotes the number of the occurrences of $a$ in $\sigma$ for every $a\in\Omega$ and every $\sigma\in\Omega^*$.
\qed
\end{theorem}

Note that in Theorem~\ref{FI}
the underlying finite probability space $P$ is \emph{quite arbitrary}, and
thus we do not impose any computability restrictions on
$P$ at all, in particular.

\subsection{Non-occurrence of an elementary event with probability zero}

From various aspects, Tadaki~\cite[Section~5.1]{T16arXiv} demonstrated the fact that
\emph{an elementary event with probability zero never occurs
in the conventional quantum mechanics}.
This fact corresponds to Theorem~\ref{thm:zero_probability} below
in the context of the theory of
\emph{operational characterization of the notion of probability},
which was developed by Tadaki~\cite{T14,T15,T16arXiv}
based on Martin-L\"of $P$-randomness.

\begin{theorem}\label{thm:zero_probability}
Let $\Omega$ be an alphabet, and let $P\in\PS(\Omega)$.
Let $a\in\Omega$.
Suppose that
$\alpha$ is Martin-L\"of $P$-random
and $P(a)=0$.
Then $\alpha$ does not contain $a$.
\qed
\end{theorem}

This result for Martin-L\"of $P$-random was, in essence,
pointed out by Martin-L\"of~\cite{M66}.
Note that
we do not impose any computability restrictions on
the underlying finite probability space~$P$ at all
in Theorem~\ref{thm:zero_probability}.
For the proof of Theorem~\ref{thm:zero_probability},
see Tadaki~\cite[Theorem~13]{T16arXiv}.
The following corollary is immediate from Theorem~\ref{thm:zero_probability}.

\begin{corollary}\label{cor:always-positive-probability}
Let $\Omega$ be an alphabet, and let $P\in\PS(\Omega)$.
For every $\alpha\in\Omega^\infty$, if $\alpha$ is Martin-L\"of $P$-random
then $P(\alpha(n))>0$ for every $n\in\N^+$.
\qed
\end{corollary}

\subsection{Conditional probability}
\label{subsec:Conditional probability}

The notion of \emph{conditional probability} in a finite probability space can be represented by
the notion of Martin-L\"of $P$-randomness in a natural manner as follows.

First, we recall
the notion of conditional probability in a finite probability space.
Let $\Omega$ be an alphabet, and let $P\in\PS(\Omega)$.
Let $B\subset\Omega$ be an event on the finite probability space~$P$.
Suppose that $P(B)>0$.
Then, for each event~$A\subset\Omega$,
the \emph{conditional probability of A given B}, denoted
$P(A|B)$, is defined as $P(A\cap B)/P(B)$.
This notion defines a finite probability space~$P_B\in\PS(B)$
by the condition that $P_B(a):=P(\{a\}|B)$ for every $a\in B$.

When an infinite sequence~$\alpha\in\Omega^\infty$ contains infinitely many elements from $B$,
$$\cond{B}{\alpha}$$
is defined as an infinite sequence in $B^\infty$ obtained from $\alpha$
by eliminating all elements
of
$\Omega\setminus B$ occurring in $\alpha$.
If $\alpha$ is Martin-L\"of $P$-random for the finite probability space~$P$ and $P(B)>0$,
then $\alpha$ contains infinitely many elements from $B$
due to Theorem~\ref{FI} above.
Therefore, $\cond{B}{\alpha}$ is properly defined in this case.
Note that the notion of $\cond{B}{\alpha}$ in our
theory
is introduced by
Tadaki~\cite{T14},
suggested by
the notion of \emph{partition}
in the theory of \emph{collectives}
introduced
by von Mises~\cite{vM64} (see Tadaki~\cite{T16arXiv} for the detail).

We can then show Theorem~\ref{conditional_probability} below, which states that
Martin-L\"of $P$-random sequences are \emph{closed under conditioning}.
For the proof of Theorem~\ref{conditional_probability}, see Tadaki~\cite[Theorem~26]{T16arXiv}.

\begin{theorem}[Closure property under conditioning, Tadaki~\cite{T14}]\label{conditional_probability}
Let $\Omega$ be an alphabet, and let $P\in\PS(\Omega)$.
Let $B\subset\Omega$ be an event on the finite probability space~$P$ with $P(B)>0$.
For every $\alpha\in\Omega^\infty$,
if $\alpha$ is Martin-L\"of $P$-random
then $\cond{B}{\alpha}$ is Martin-L\"of $P_B$-random for the finite probability space~$P_B$.
\qed
\end{theorem}

Note that in Theorem~\ref{conditional_probability}
the underlying finite probability space $P$ is \emph{quite arbitrary}, and
thus we do not impose any computability restrictions on $P$ at all, in particular.

\subsection{\boldmath Independence of Martin-L\"of $P$-random infinite sequences}\label{sec-IMLP}

Tadaki~\cite{T15} proposed
the notion of \emph{independence}
of Martin-L\"of $P$-random infinite sequences.
This notion is introduced
in the following manner:
Let $\Omega_1,\dots,\Omega_K$ be alphabets.
For any $\alpha_1\in\Omega_1^\infty,\dots,\alpha_K\in\Omega_K^\infty$,
we use
$$\alpha_1\times\dots\times\alpha_K$$
to denote an infinite sequence~$\alpha$ over $\Omega_1\times\dots\times\Omega_K$
such that
$$\alpha(n)=(\alpha_1(n),\dots,\alpha_K(n))$$
for every $n\in\N^+$.
On the other hand, for any $P_1\in\PS(\Omega_1),\dots,P_K\in\PS(\Omega_K)$,
we use
$$P_1\times\dots\times P_K$$
to denote a finite probability space~$Q\in\PS(\Omega_1\times\dots\times\Omega_K)$ such that
$$Q(a_1,\dots,a_K)=P_1(a_1)\dotsm P_K(a_K)$$
for every $a_1\in\Omega_1,\dotsc,a_K\in\Omega_K$.

\begin{definition}[Independence of Martin-L\"of $P$-random infinite sequences, Tadaki~\cite{T15}]
\label{independency-of-ensembles}
Let $\Omega_1,\dotsc,\Omega_K$ be alphabets, and let
$P_1\in\PS(\Omega_1),\dots,P_K\in\PS(\Omega_K)$.
For each $k=1,\dots,K$,
let $\alpha_k$ be a Martin-L\"of $P_k$-random infinite sequence over $\Omega_k$.
We say that $\alpha_1,\dots,\alpha_K$ are
\emph{independent}
if $\alpha_1\times\dots\times\alpha_K$ is Martin-L\"of $P_1\times\dots\times P_K$-random.
\qed
\end{definition}

Note that the notion of the independence of Martin-L\"of $P$-random infinite sequences
is introduced by
Tadaki~\cite{T15},
suggested by
the notion of \emph{independence} of \emph{collectives}
in the theory of collectives
introduced
by von Mises~\cite{vM64} (see Tadaki~\cite{T16arXiv} for the detail).
In Definition~\ref{independency-of-ensembles}
the underlying finite probability spaces $P_1,\dots,P_K$ are \emph{quite arbitrary}, and thus
we do not impose any computability restrictions on
$P_1,\dots,P_K$ at all, in particular.

\subsection{Marginal distribution}

The following two theorems together state that
Martin-L\"of $P$-randomness is
\emph{closed under the formation of marginal distribution}.
See Tadaki~\cite[Theorem~9.4]{T18arXiv} for the proof
of Theorem~\ref{contraction2}.

\begin{theorem}[Tadaki~\cite{T18arXiv}]\label{contraction2}
Let $\Omega$ and $\Theta$ be alphabets.
Let $P\in\PS(\Omega\times\Theta)$, and let $\alpha\in(\Omega\times\Theta)^\infty$.
Suppose that $\beta$ is an infinite sequence over $\Omega$ obtained from $\alpha$
by replacing each element $(m,l)$ in $\alpha$ by $m$.
If $\alpha$ is Martin-L\"of $P$-random then $\beta$ is Martin-L\"of $Q$-random,
where $Q\in\PS(\Omega)$ such that
$$Q(m):=\sum_{l\in\Theta}P(m,l)$$
for every $m\in\Omega$.
\qed
\end{theorem}

In a similar manner to the proof of Theorem~\ref{contraction2},
we can prove Theorem~\ref{contraction2-2} below, a symmetrical result to Theorem~\ref{contraction2}.

\begin{theorem}\label{contraction2-2}
Let $\Omega$ and $\Theta$ be alphabets.
Let $P\in\PS(\Omega\times\Theta)$, and let $\alpha\in(\Omega\times\Theta)^\infty$.
Suppose that $\gamma$ is an infinite sequence over $\Theta$ obtained from $\alpha$
by replacing each element $(m,l)$ in $\alpha$ by $l$.
If $\alpha$ is Martin-L\"of $P$-random then $\gamma$ is Martin-L\"of $R$-random,
where $R\in\PS(\Theta)$ such that
$$R(l):=\sum_{m\in\Omega}P(m,l)$$
for every $l\in\Theta$.
\qed
\end{theorem}

Note that
in both Theorem~\ref{contraction2} and Theorem~\ref{contraction2-2}
the underlying finite probability space $P$ on $\Omega\times\Theta$
is \emph{quite arbitrary}, and
thus we do not impose any computability restrictions on $P$ at all, in particular.
Then using Theorems~\ref{contraction2} and \ref{contraction2-2}
we can show the following theorem.

\begin{theorem}\label{independency-imlplies-each-randomness}
Let $\Omega_1,\dotsc,\Omega_K$ be alphabets, and let
$P_1\in\PS(\Omega_1),\dots,P_K\in\PS(\Omega_K)$.
Let $\alpha_1\in\Omega_1^\infty,\dots,\alpha_K\in\Omega_K^\infty$.
Suppose that $\alpha_1\times\dots\times\alpha_K$ is Martin-L\"of $P_1\times\dots\times P_K$-random.
Then $\alpha_k$ is Martin-L\"of $P_k$-random for every $k=1,\dots,K$.
\end{theorem}

\begin{proof}
Note that $P_1\times\dots\times P_K$ and $\alpha_1\times\dots\times\alpha_K$
can be regarded as $(P_1\times\dots\times P_{K-1})\times P_K$ and
$(\alpha_1\times\dots\times\alpha_{K-1})\times\alpha_K$, respectively.
Thus, it suffices to prove the result in the case of $K=2$.

Now, consider the case of $K=2$.
Let $\alpha_1\in\Omega_1^\infty$ and $\alpha_2\in\Omega_2^\infty$.
Suppose that $\alpha_1\times\alpha_2$ is Martin-L\"of $P_1\times P_2$-random.
On the one hand,
the infinite sequence $\alpha_1$ is obtained from $\alpha_1\times\alpha_2$
by replacing each element $(m,l)$ in $\alpha_1\times\alpha_2$ by $m$.
On the other hand, note that
\[
  \sum_{l\in\Omega_2}(P_1\times P_2)(m,l)=\sum_{l\in\Omega_2}P_1(m)P_2(l)=P_1(m)\sum_{l\in\Omega_2}P_2(l)=P_1(m)
\]
for each $m\in\Omega_1$.
Thus, it follows from Theorems~\ref{contraction2} that
$\alpha_1$ is Martin-L\"of $P_1$-random.

Similarly, on the one hand,
the infinite sequence $\alpha_2$ is obtained from $\alpha_1\times\alpha_2$
by replacing each element $(m,l)$ in $\alpha_1\times\alpha_2$ by $l$.
On the other hand, note that
\[
  \sum_{m\in\Omega_1}(P_1\times P_2)(m,l)=P_2(l)
\]
for every $l\in\Omega_2$.
Thus, it follows from Theorems~\ref{contraction2-2} that
$\alpha_2$ is Martin-L\"of $P_2$-random.
This completes the proof.
\end{proof}

Note that in Theorems~\ref{independency-imlplies-each-randomness}
the underlying finite probability spaces $P_1,\dots,P_K$
are \emph{quite arbitrary}, and
thus we do not impose any computability restrictions on $P_1,\dots,P_K$ at all, in particular.

\section{Postulates of quantum mechanics}
\label{QM}

In this section, we review the central postulates of (the conventional) quantum mechanics.
Specifically, we here refer to the postulates of
the \emph{conventional}
quantum mechanics
in the form presented in Nielsen and Chuang \cite[Chapter 2]{NC00},
with slight modification.
Note that their original postulates were given as postulates of
(the conventional)
quantum mechanics for a finite-dimensional quantum system, i.e.,
a quantum system whose state space is a finite-dimensional complex Hilbert space,
since their book is a textbook of the field of quantum computation and quantum information
where finite-dimensional quantum systems are typical.
In contrast,
we do not impose any restrictions on the dimensionality of the underlying state spaces
in the central postulates of (the conventional) quantum mechanics of the form presented in what follows.

The first postulate of quantum mechanics is about \emph{state space} and \emph{state vector}.

\begin{postulate}[State space and state vector]\label{state_space}
Associated to any isolated physical system is a separable complex Hilbert space
known as the \emph{state space} of the system.
The system is completely described by its \emph{state vector},
which is a unit vector in the system's state space.
\qed
\end{postulate}

The second postulate of quantum mechanics is about the \emph{composition} of systems.

\begin{postulate}[Composition of systems]\label{composition}
The state space of a composite physical system is the tensor product of the state spaces of the component physical systems.
Moreover, if we have systems numbered $1$ through $n$, and system number~$i$ is
prepared
in the state~$\ket{\Psi_i}$,
then the joint state of the total system is
$\ket{\Psi_1}\otimes\ket{\Psi_2}\otimes\dots\otimes\ket{\Psi_n}$.
\qed
\end{postulate}

The third postulate of quantum mechanics is about the \emph{time-evolution} of
\emph{closed} quantum systems.

\begin{postulate}[Unitary time-evolution]\label{evolution}
The evolution of a \emph{closed} quantum system is described by a \emph{unitary transformation}.
That is,
the state~$\ket{\Psi_1}$ of the system at time~$t_1$ is related to the state~$\ket{\Psi_2}$ of the system at time~$t_2$
by a unitary operator~$U$
on the state space,
which depends only on the times~$t_1$ and $t_2$,
in such a way that
$\ket{\Psi_2}=U\ket{\Psi_1}$.
\qed
\end{postulate}

The forth postulate of quantum mechanics is about \emph{measurements} on quantum systems.

\begin{postulate}[Quantum measurement]\label{Born-rule}
A quantum measurement is described by a collection $\{M_m\}_{m\in\Omega}$ of \emph{measurement operators}
which is a finite collection of bounded operators on the state space of the system being measured,
satisfying the \emph{completeness equation}
\begin{equation*}
  \sum_{m\in\Omega} M_m^\dag M_m=I,
\end{equation*}
where $M_m^\dag$ denotes the adjoint of the bounded operator $M_m$, and $I$ denotes the identity operator on the state space.
\begin{enumerate}
\item
The set of possible outcomes of the measurement
equals
the
non-empty
finite set $\Omega$.
If the state of the quantum system is $\ket{\Psi}$ immediately before the measurement then
the probability that result~$m$ occurs is given by
\begin{equation*}
  \bra{\Psi}M_m^\dag M_m\ket{\Psi}.
\end{equation*}
\item
Given that outcome $m$ occurred,
the state of the quantum system immediately after the measurement is
\begin{equation*}
  \frac{M_m\ket{\Psi}}{\sqrt{\bra{\Psi}M_m^\dag M_m\ket{\Psi}}}.
\end{equation*}
\qed
\end{enumerate}
\end{postulate}

Postulate~\ref{Born-rule}~(i) is the so-called \emph{Born rule}, i.e,
\emph{the probability interpretation of the wave function}.
Thus,
the Born rule, Postulate~\ref{Born-rule}~(i),
uses the \emph{notion of probability}.
However, the operational characterization of the notion of probability is not given in the Born rule,
and therefore the relation of its statement to a specific infinite sequence of outcomes of quantum measurements which are being generated by an infinitely repeated measurements is \emph{unclear}.
Tadaki~\cite{T14CCR,T15Kokyuroku,T16CCR,T16QIP,T16QIT35,T17SCIS,T18arXiv} fixed this point.

In this paper
as well as our former works~\cite{T14CCR,T15Kokyuroku,T16CCR,T16QIP,T16QIT35,T17SCIS,T18arXiv},
\emph{we keep Postulates~\ref{state_space}, \ref{composition}, and \ref{evolution}
in their original forms without any modifications}.
The principle of typicality, Postulate~\ref{POT} below,
is proposed as a \emph{refinement} of Postulate~\ref{Born-rule} to replace it,
based on the notion of \emph{Martin-L\"of randomness with respect to a probability measure}.

\section{The principle of typicality}
\label{MWI}

Everett~\cite{E57} introduced
the \emph{many-worlds interpretation of quantum mechanics} (\emph{MWI}, for short)
in 1957.
Actually, his MWI was more than just an interpretation of quantum mechanics.
It aimed to derive
Postulate~\ref{Born-rule}~(i), the Born rule,
from the remaining postulates,
i.e.,
from Postulates~\ref{state_space}, \ref{composition}, and \ref{evolution} above.
In this sense, Everett~\cite{E57} proposed his MWI
as a ``metatheory'' of quantum mechanics.
The point is that in MWI the measurement process is fully treated
as the interaction between a system being measured and an apparatus measuring it, based only on
Postulates~\ref{state_space}, \ref{composition}, and \ref{evolution}.
Then his MWI tried to derive Postulate~\ref{Born-rule}~(i) in such a setting.
However, his attempt does not seem successful.
One of the reasons is thought to be its mathematical vagueness.

Tadaki~\cite{T16CCR}
reformulated the original framework of MWI by Everett~\cite{E57}
\emph{in a form of mathematical rigor} from a modern point of view.
The point of this rigorous treatment of the framework of MWI
is the use of the notion of \emph{probability measure representation} and
its \emph{induction of probability measure}, as presented in Section~\ref{MLRam}.
Tadaki~\cite{T16CCR} then introduced the \emph{principle of typicality}
in this rigorous framework.
The principle of typicality is
an
operational refinement of Postulate~\ref{Born-rule}
based on
the notion of Martin-L\"{o}f randomness with respect to a probability measure
given in Definition~\ref{ML-randomness-wrtm}.
In what follows,
we review this rigorous framework of MWI based the principle of typicality.
See Tadaki~\cite{T18arXiv}
for the details of this
rigorous
framework of MWI
and its validity.
Also, for the details of
the difference between
this
framework of MWI
and that of the original MWI~\cite{E57}, see Tadaki~\cite{T18arXiv}.
The contributions of this paper presented
in Sections~\ref{sec-QM-equality}
and
\ref{sec:GHZ-QM}
demonstrate further
the validity of this rigorous framework of MWI
based the principle of typicality.

Now,
according to Tadaki~\cite{T18arXiv},
let us introduce the
setting
of MWI \emph{in terms of our terminology
in a form of mathematical rigor}.
Let $\mathcal{S}$ be an arbitrary quantum system with state space $\mathcal{H}$
of an arbitrary dimension.
Consider a measurement over $\mathcal{S}$ described
by arbitrary measurement operators $\{M_m\}_{m\in\Omega}$ satisfying
the completeness equation,
\begin{equation}\label{completeness-equation}
  \sum_{m\in\Omega}M_m^{\dagger} M_m=I.
\end{equation}
Here, $\Omega$
is the set of all possible outcomes of the measurement.%
\footnote{The set $\Omega$ is non-empty and finite, and therefore it is an alphabet.}
Let $\mathcal{A}$ be an apparatus performing the measurement described by $\{M_m\}_{m\in\Omega}$,
which is a quantum system with state space $\ssoa$.
According to Postulates~\ref{state_space}, \ref{composition}, and \ref{evolution},
the measurement process of
the measurement described by the measurement operators $\{M_m\}_{m\in\Omega}$
is described by a unitary operator $U$ such that
\begin{equation}\label{single_measurement}
  U(\ket{\Psi}\otimes\ket{\Phi_{\mathrm{init}}})=\sum_{m\in\Omega}(M_m\ket{\Psi})\otimes\ket{\Phi[m]}
\end{equation}
for every $\ket{\Psi}\in\mathcal{H}$.
This is due to von Neumann~\cite[Section~VI.3]{vN55}
in the case where the measurement operators $\{M_m\}_{m\in\Omega}$ form
a projection-valued measure
such that $M_m$ is of rank $1$, i.e., $\dim M_m(\mathcal{H})=1$, for every $m\in\Omega$.
For the derivation of \eqref{single_measurement} in the general case,
see Tadaki~\cite[Section~7.1]{T18arXiv}.
Actually,
$U$ describes the interaction between the system $\mathcal{S}$ and the apparatus $\mathcal{A}$.
The vector $\ket{\Phi_{\mathrm{init}}}\in\ssoa$ is
the initial state of the apparatus $\mathcal{A}$, and $\ket{\Phi[m]}\in\ssoa$ is
a final state of the apparatus $\mathcal{A}$ for each $m\in\Omega$,
with $\braket{\Phi[m]}{\Phi[m']}=\delta_{m,m'}$.
For every $m\in\Omega$, the state $\ket{\Phi[m]}$ indicates that
\emph{the apparatus $\mathcal{A}$ records the value $m$
as the measurement outcome}.
By the unitary interaction~\eqref{single_measurement} as a measurement process,
a correlation (i.e., entanglement) is generated between
the system and the apparatus.

In the framework of MWI,
we consider countably infinite copies of the system $\mathcal{S}$,
and consider a countably infinite repetition of the measurements
described by the identical measurement operators $\{M_m\}_{m\in\Omega}$
performed over each of such copies in a sequential order,
where each of the measurements is described
by the unitary time-evolution~\eqref{single_measurement}.
As repetitions of the measurement progressed,
correlations between the systems and the apparatuses are being generated in sequence
in the superposition of the total system consisting of the systems and the apparatuses.
The detail
is described as follows.

For each $n\in\N^+$, let $\mathcal{S}_n$ be the $n$th copy of the system $\mathcal{S}$ and
$\mathcal{A}_n$ the $n$th copy of the apparatus $\mathcal{A}$.
Each $\mathcal{S}_n$ is prepared in a state $\ket{\Psi_n}$
while all $\mathcal{A}_n$ are prepared in an identical state $\ket{\Phi_{\mathrm{init}}}$.
The measurement
described by the measurement operators $\{M_m\}_{m\in\Omega}$
is performed over each $\mathcal{S}_n$ one by one
in the increasing order of $n$,
by interacting each $\mathcal{S}_n$ with $\mathcal{A}_n$
according to the unitary time-evolution~\eqref{single_measurement}.
For each $n\in\N^+$,
let $\mathcal{H}_n$ be the state space of the total system consisting of
the first $n$ copies
$\mathcal{S}_1, \mathcal{A}_1, \mathcal{S}_2, \mathcal{A}_2,\dots,\mathcal{S}_n, \mathcal{A}_n$
of the system $\mathcal{S}$ and the apparatus $\mathcal{A}$.
These successive interactions between the copies of
the system $\mathcal{S}$ and the apparatus $\mathcal{A}$ as measurement processes
proceed in the following manner:

The
starting
state of the total system,
which consists of $\mathcal{S}_1$ and $\mathcal{A}_1$,
is $\ket{\Psi_1}\otimes\ket{\Phi_{\mathrm{init}}}\in\mathcal{H}_1$.
Immediately after the measurement
described by $\{M_m\}_{m\in\Omega}$
over $\mathcal{S}_1$,
the total system results in the state
\begin{align*}
  \sum_{m_1\in\Omega} (M_{m_1}\ket{\Psi_1})\otimes\ket{\Phi[m_1]}\in\mathcal{H}_1
\end{align*}
by the interaction~\eqref{single_measurement} as a measurement process.
In general,
immediately before
performing
the measurement
described by $\{M_m\}_{m\in\Omega}$
over $\mathcal{S}_n$,
the state of the total system,
which consists of
$\mathcal{S}_1, \mathcal{A}_1, \mathcal{S}_2, \mathcal{A}_2,\dots,\mathcal{S}_n, \mathcal{A}_n$,
is
\begin{align*}
  \sum_{m_1,\dots,m_{n-1}\in\Omega}
  (M_{m_1}\ket{\Psi_1})\otimes\dots\otimes(M_{m_{n-1}}\ket{\Psi_{n-1}})\otimes\ket{\Psi_n}
  \otimes\ket{\Phi[m_1]}\otimes\dots\otimes\ket{\Phi[m_{n-1}]}\otimes\ket{\Phi_{\mathrm{init}}}
\end{align*}
in $\mathcal{H}_n$,
where
$\ket{\Psi_n}$ is the initial state of $\mathcal{S}_n$ and
$\ket{\Phi_{\mathrm{init}}}$ is the initial state of $\mathcal{A}_n$.
Immediately after the measurement
described by $\{M_m\}_{m\in\Omega}$
over $\mathcal{S}_n$,
the total system results in the state
\begin{align}
   &\sum_{m_1,\dots,m_{n}\in\Omega}
  (M_{m_1}\ket{\Psi_1})\otimes\dots\otimes(M_{m_{n}}\ket{\Psi_n})
  \otimes\ket{\Phi[m_1]}\otimes\dots\otimes\ket{\Phi[m_{n}]} \nonumber \\ %
  =&\sum_{m_1,\dots,m_{n}\in\Omega}
  (M_{m_1}\ket{\Psi_1})\otimes\dots\otimes(M_{m_{n}}\ket{\Psi_n})
  \otimes\ket{\Phi[m_1\dots m_{n}]} \label{total_system}
\end{align}
in $\mathcal{H}_n$,
by the interaction~\eqref{single_measurement} as a measurement process between
the system $\mathcal{S}_n$ prepared in the state $\ket{\Psi_n}$
and the apparatus $\mathcal{A}_n$ prepared in the state $\ket{\Phi_{\mathrm{init}}}$.
The vector $\ket{\Phi[m_1\dots m_n]}$ denotes
the vector $\ket{\Phi[m_1]}\otimes\dots\otimes\ket{\Phi[m_n]}$ which represents
the state of $\mathcal{A}_1,\dots,\mathcal{A}_n$.
This state indicates that \emph{the apparatuses $\mathcal{A}_1,\dots,\mathcal{A}_n$ record
the values $m_1\dots m_n$
as the measurement outcomes
over $\mathcal{S}_1,\dots,\mathcal{S}_n$, respectively}.

In the superposition~\eqref{total_system},
on letting $n\to\infty$,
the length of
the records $m_1\dots m_n$ of the values
as the measurement outcomes
in the apparatuses $\mathcal{A}_1,\dots,\mathcal{A}_n$ diverges to infinity.
The
consideration of
this
limiting case results in the
definition of a \emph{world}.
Namely, a \emph{world} is defined as
an
infinite sequence of records of the values
as the measurement outcomes
in the apparatuses.
Thus, in the case described
so far,
a world is an infinite sequence over $\Omega$,
and the finite records $m_1\dots m_n$ in each state $\ket{\Phi[m_1\dots m_n]}$
in the superposition~\eqref{total_system} of the total system is a \emph{prefix} of a world.

For aiming at deriving Postulate~\ref{Born-rule}, the original MWI~\cite{E57} assigned
``weight'' to each of worlds.
In our rigorous framework of MWI,
we introduce a \emph{probability measure} on the set of all worlds in the following manner:
First, we introduce a probability measure representation on the set of prefixes of worlds, i.e.,
the set $\Omega^*$ in this case.
This probability measure representation is given by a function $r\colon\Omega^*\to[0,1]$ with
\begin{equation}\label{rpmwi}
  r(m_1\dotsc m_n)=\prod_{k=1}^n\bra{\Psi_k}M_{m_k}^\dag M_{m_k}\ket{\Psi_k},
\end{equation}
which is the square of the norm of each vector
$(M_{m_1}\ket{\Psi_1})\otimes\dots\otimes(M_{m_{n}}\ket{\Psi_n})\otimes\ket{\Phi[m_1\dots m_{n}]}$
in the superposition~\eqref{total_system}.
Using the completeness equation~\eqref{completeness-equation},
it is easy to check that $r$ is certainly a probability measure representation over $\Omega$.
We call the probability measure representation $r$
\emph{the
probability
measure representation
for the prefixes of
worlds}.
We then adopt
the probability measure \emph{induced} by the probability measure representation $r$
for the prefixes of worlds
as \emph{the probability measure on the set of all worlds}.

We summarize the above consideration and clarify
the definitions of the notion of \emph{world} and
the notion of \emph{the
probability
measure representation for the prefixes of worlds}
as in the following.

\begin{definition}%
[The
probability
measure representation for the prefixes of worlds, Tadaki~\cite{T18arXiv}]\label{pmrpwst}
Consider an arbitrary
quantum system $\mathcal{S}$ and
a
measurement over $\mathcal{S}$
described by arbitrary \emph{measurement operators $\{M_m\}_{m\in\Omega}$},
where the measurement process is described by \eqref{single_measurement} as an interaction
of the system $\mathcal{S}$ with an apparatus $\mathcal{A}$.
We suppose the following situation:
\begin{enumerate}
\item There are countably infinite copies $\mathcal{S}_1, \mathcal{S}_2, \mathcal{S}_3 \dotsc$
  of the system $\mathcal{S}$ and countably infinite copies
  $\mathcal{A}_1, \mathcal{A}_2, \mathcal{A}_3, \dotsc$ of the apparatus $\mathcal{A}$.
\item For each $n\in\N^+$, the system $\mathcal{S}_n$ is prepared in a state $\ket{\Psi_n}$,%
  \footnote{In Definition~\ref{pmrpwst}, all $\ket{\Psi_n}$ are not required to be an identical state.
  In the applications of the principle of typicality in this paper,
  which are presented in Sections~\ref{sec-QM-equality} and \ref{sec:GHZ-QM},
  all $\ket{\Psi_n}$ are chosen to be an identical state, however.}
  while the apparatus $\mathcal{A}_n$ is prepared in a state $\ket{\Phi_{\mathrm{init}}}$,
  and then the measurement described by $\{M_m\}_{m\in\Omega}$ is performed over $\mathcal{S}_n$
  by interacting it with the apparatus $\mathcal{A}_n$ according to
  the unitary time-evolution~\eqref{single_measurement}.
\item Starting the measurement described by $\{M_m\}_{m\in\Omega}$ over $\mathcal{S}_1$,
  the measurement described by $\{M_m\}_{m\in\Omega}$ over each $\mathcal{S}_n$ is
  performed in the increasing order of $n$.
\end{enumerate}
We then note that,
for each $n\in\N^+$,
immediately after the measurement described by $\{M_m\}_{m\in\Omega}$ over $\mathcal{S}_n$,
the state of the total system consisting of
$\mathcal{S}_1, \mathcal{A}_1, \mathcal{S}_2, \mathcal{A}_2,\dots,\mathcal{S}_n, \mathcal{A}_n$ is
$$\ket{\Theta_n}:=\sum_{m_1,\dots,m_n\in\Omega}\ket{\Theta(m_1,\dots,m_n)},$$
where
$$\ket{\Theta(m_1,\dots,m_n)}
:=(M_{m_1}\ket{\Psi_1})\otimes\dots\otimes(M_{m_{n}}\ket{\Psi_n})
\otimes\ket{\Phi[m_1]}\otimes\dots\otimes\ket{\Phi[m_n]}.$$
The vectors $M_{m_1}\ket{\Psi_1},\dots,M_{m_{n}}\ket{\Psi_n}$ are states of
$\mathcal{S}_1,\dots,\mathcal{S}_n$, respectively, and
the vectors
$$\ket{\Phi[m_1]},\dots,\ket{\Phi[m_n]}$$
are
states of $\mathcal{A}_1,\dots,\mathcal{A}_n$, respectively.
The state
vector
$\ket{\Theta_n}$ of the total system is normalized while
each of the vectors $\{\ket{\Theta(m_1,\dots,m_n)}\}_{m_1,\dots,m_n\in\Omega}$
is not necessarily normalized.
Then,
\emph{the
probability
measure representation for the prefixes of
worlds}
is defined
as
a
function $p\colon\Omega^*\to[0,1]$ such that
\begin{equation}\label{p=bTmkT}
  p(m_1\dotsc m_n)=\braket{\Theta(m_1,\dots,m_n)}{\Theta(m_1,\dots,m_n)}
\end{equation}
for every $n\in\N^+$ and every $m_1,\dots,m_n\in\Omega$.
Moreover,
an infinite sequence over $\Omega$,
i.e., an infinite sequence of
possible
outcomes of the measurement described by $\{M_m\}_{m\in\Omega}$,
is called a \emph{world}.
\qed
\end{definition}

In Definition~\ref{pmrpwst}, it is easy to check that the function $p$ defined by \eqref{p=bTmkT} is
certainly a probability measure representation over $\Omega$.

As mentioned above, the original MWI by Everett~\cite{E57} aimed to derive
Postulate~\ref{Born-rule}~(i), the Born rule,
from the remaining postulates.
However, it would seem impossible to do this for several reasons
(see Tadaki~\cite{T18arXiv} for
these reasons).
Instead, it is appropriate to introduce an additional postulate
in our rigorous framework of MWI developed above,
in order to overcome the defect of the original MWI
and to make quantum mechanics \emph{operationally perfect}.
Thus, we put forward Postulate~\ref{POT} below, the \emph{principle of typicality}.

\begin{postulate}[The principle of typicality, Tadaki~\cite{T16CCR}]\label{POT}
Our world is \emph{typical}.
Namely, our world is Martin-L\"of random with respect to the probability measure on the set of all worlds,
induced by
the
probability
measure representation
for the prefixes of
worlds,
in the superposition of the total system
which consists of
systems being measured and apparatuses measuring them.
\qed
\end{postulate}

For the comprehensive arguments of the validity of
Postulate~\ref{POT}, the principle of typicality,
see Tadaki~\cite{T18arXiv}.
For example, based on the results of
Tadaki~\cite{T14,T15,T16arXiv},
we can see that Postulate~\ref{POT} is certainly a refinement of
Postulate~\ref{Born-rule}
from the point of view of our intuitive understanding of the notion of probability.

\section{Refinement of the argument of quantum mechanics to violate Bell's inequality,
based on the principle of typicality}
\label{sec-QM-equality}

In this section,
based on the \emph{principle of typicality},
we \emph{refine} and \emph{reformulate}
that argument
of quantum mechanics
to violate Bell's inequality which is 
given in Section~2.6 of Nielsen and Chuang~\cite{NC00}.
Thus,
in what follows,
according to Nielsen and Chuang~\cite[Section~2.6]{NC00},
we investigate
Protocol~\ref{Bell} below
due to Bell~\cite{Bell64}, Clauser, et al.~\cite{CHSH69}, and Nielsen and Chuang~\cite{NC00}
in \emph{our framework of quantum mechanics based on the principle of typicality}, i.e.,
in our rigorous framework of MWI based on the principle of typicality,
developed in the preceding section.

First, we fix some notation.
Let $\ket{0}$ and $\ket{1}$ be an orthonormal basis of the state space of a single qubit system.
Based on them we define a state $\ket{+}$ of a system of a single qubit by
\begin{equation}\label{eq:ket-plus}
  \ket{+}:=\frac{\ket{0}+\ket{1}}{\sqrt{2}},
\end{equation}
and define
the \emph{Bell state}
$\ket{\beta_{11}}$ of a system of two qubits by
\begin{equation}\label{Bell-fps-calculation:eq01}
  \ket{\beta_{11}}:=\frac{\ket{0}\otimes\ket{1}-\ket{1}\otimes\ket{0}}{\sqrt{2}}.
\end{equation}
Pauli matrices $X,Y,Z$ are defined by
\begin{equation}\label{Bell-fps-calculation:eq02}
X:=\ket{1}\bra{0}+\ket{0}\bra{1},\quad
Y:=i\ket{1}\bra{0}-i\ket{0}\bra{1},\quad
Z:=\ket{0}\bra{0}-\ket{1}\bra{1}.
\end{equation}
We deal with four observables $R,Q,S,T$ of
a system of a single qubit
defined
by the following way:%
\footnote{We intentionally put $R$ before $Q$ at variance with Alphabetical order,
but we use the same notation
exactly
as in Section~2.6 of Nielsen and Chuang~\cite{NC00}.
Originally, the observables $Q$ and $R$ are defined via (2.227) and (2.228) of Nielsen and Chuang~\cite{NC00} in their book.
However, our result~\eqref{Pcdmn=1o161+-1cdmns2-Bell} below reveals that
it is natural to define these $Q$ and $R$ in reverse,
unlike the definitions~(2.227) and (2.228) of their book.}
\begin{equation}\label{Bell-fps-calculation:eq03}
R:=X,\quad
Q:=Z,\quad
S:=-\frac{1}{\sqrt{2}}X-\frac{1}{\sqrt{2}}Z,\quad
T:=-\frac{1}{\sqrt{2}}X+\frac{1}{\sqrt{2}}Z.
\end{equation}
For each $n\ge 2$,
we denote the identity operator on the state space of a system of $n$ qubits by $I_{2^n}$.
Thus, for example, $I_2=\product{0}{0}+\product{1}{1}$ and $I_{16}=I_2\otimes I_2\otimes I_2\otimes I_2$.

\begin{protocol}\label{Bell}
The protocol involves three parties, Charlie, Alice, and Bob.
They together repeat the following procedure \emph{forever}.
{
\renewcommand{\labelenumi}{Step \arabic{enumi}:}
\setlength{\leftmargini}{40pt} 
\begin{enumerate}
\item Charlie prepares a quantum system of two qubits in the state
$\ket{\beta_{11}}$.
\item
Charlie
passes the first qubit to Alice, and the second qubit to Bob.
\end{enumerate}
}

Then Alice and Bob do the following respectively.
On the one hand, Alice does the following:
{
\renewcommand{\labelenumi}{Step A\arabic{enumi}:}
\setlength{\leftmargini}{48pt}
\begin{enumerate}
\setcounter{enumi}{2}
\item Alice tosses a fair coin $C$ to get outcome $c\in\{0,1\}$.
\item Alice performs the measurement of either $R$ or $Q$
  over the first qubit to obtain outcome $m\in\{+1,-1\}$,
  depending on $c=0$ or $1$.
\end{enumerate}
}

On the other hand, Bob does the following:
{
\renewcommand{\labelenumi}{Step B\arabic{enumi}:}
\setlength{\leftmargini}{48pt}
\begin{enumerate}
\setcounter{enumi}{2}
\item Bob tosses a fair coin $D$ to get outcome $d\in\{0,1\}$.
\item Bob performs the measurement of either $S$ or $T$
  over the second qubit to obtain outcome $n\in\{+1,-1\}$,
  depending on $d=0$ or $1$.
\end{enumerate}
\renewcommand{\labelenumi}{(\roman{enumi})}
}

While repeating the above procedure forever,
Alice and Bob together calculate
the conditional averages~$\langle RS\rangle$, $\langle QS\rangle$, $\langle RT\rangle$,
and
$\langle QT\rangle$,
which are defined by the equations~\eqref{def-conditional-averages-QM} or
the equations~\eqref{def-conditional-averages-ALR} below
and whose meaning is explained below.
\qed
\end{protocol}

In what follows,
we analyze Protocol~\ref{Bell}
in our rigorous framework of quantum mechanics based on
Postulate~\ref{POT}, the principle of typicality,
together with Postulates~\ref{state_space}, \ref{composition}, and \ref{evolution}.
To complete this,
\emph{we have to implement everything in
Steps~A3 and A4 and Steps~B3 and B4
of Protocol~\ref{Bell}
by unitary time-evolution}.

\subsection{Unitary implementation of all steps done by Alice and Bob}

We denote the system of the first qubit which Charlie sends to Alice in Step~2
by $\mathcal{Q}_{\mathrm{A}}$
with
state space $\mathcal{H}_{\mathrm{A}}$,
and denote the system of the second qubit which Charlie sends to Bob in Step~2
by $\mathcal{Q}_{\mathrm{B}}$
with
state space $\mathcal{H}_{\mathrm{B}}$.
We analyze
Steps~A3 and A4 and Steps~B3 and B4
of Protocol~\ref{Bell}
in our framework of quantum mechanics based on the principle of typicality.
In particular,
we realize each of
the coin tossings
in Steps~A3 and B3
by a measurement of the observable $\ket{1}\bra{1}$ over
a system of a single qubit in the state $\ket{+}$.
We
will
then
describe all the measurement processes during
Steps~A3 and A4 and Steps~B3 and B4
as a single unitary interaction between systems and apparatuses.

On the one hand,
each of Steps~A3 and A4 done by Alice is implemented by a unitary time-evolution
in the following manner:

\paragraph{Unitary implementation of Step~A3 done by Alice.}

To realize the coin tossing
by Alice
in Step~A3 of Protocol~\ref{Bell}
we make use of a measurement over a
single
qubit system.
Namely, to implement
Step~A3
we
introduce
a
single
qubit system $\mathcal{Q}_\mathrm{A3}$ with state space $\mathcal{H}_\mathrm{A3}$,
and perform a measurement over the system $\mathcal{Q}_\mathrm{A3}$ described
by a unitary time-evolution:
$$U_\mathrm{A3}(\ket{c}\otimes\ket{\Phi_\mathrm{A3}^{\mathrm{init}}})
=\ket{c}\otimes\ket{\Phi_\mathrm{A3}[c]}$$
for every $c\in\{0,1\}$,
where $\ket{c}\in\mathcal{H}_\mathrm{A3}$.
The vector $\ket{\Phi_\mathrm{A3}^{\mathrm{init}}}$ is the initial state of an apparatus
$\mathcal{A}_\mathrm{A3}$
measuring $\mathcal{Q}_\mathrm{A3}$,
and $\ket{\Phi_\mathrm{A3}[c]}$ is a final state of
the apparatus
$\mathcal{A}_\mathrm{A3}$
for each $c\in\{0,1\}$.%
\footnote{We assume, of course, the orthogonality of the final states $\ket{\Phi_\mathrm{A3}[c]}$, i.e.,
the property that
$\braket{\Phi_\mathrm{A3}[c]}{\Phi_\mathrm{A3}[c']}=\delta_{c,c'}$.
Furthermore,
we assume the orthogonality of the finial states for
each of all apparatuses which appear in the rest of this section.}
Prior to the measurement,
the system $\mathcal{Q}_\mathrm{A3}$ is prepared in the state $\ket{+}\in\mathcal{H}_\mathrm{A3}$.

\paragraph{Unitary implementation of Step~A4 done by Alice.}

Let $\{E^{\mathrm{A}}_{0,m}\}_{m=\pm 1}$ and $\{E^{\mathrm{A}}_{1,m}\}_{m=\pm 1}$ be
the PVMs in $\mathcal{H}_{\mathrm{A}}$ corresponding to the observables $R$ and $Q$, respectively.
Namely, let
\begin{equation}\label{Bell-fps-calculation:eq-A01}
  R=E_{0,+1}^{\mathrm{A}}-E_{0,-1}^{\mathrm{A}}\quad\text{ and }\quad
  Q=E_{1,+1}^{\mathrm{A}}-E_{1,-1}^{\mathrm{A}}
\end{equation}
be the spectral decompositions of $R$ and $Q$, respectively,
where $\{E^{\mathrm{A}}_{0,m}\}_{m=\pm 1}$ and $\{E^{\mathrm{A}}_{1,m}\}_{m=\pm 1}$ are
collections of projectors
such that
\begin{equation*}
  E_{0,+1}^{\mathrm{A}}E_{0,-1}^{\mathrm{A}}=0\quad\text{ and }\quad
  E_{1,+1}^{\mathrm{A}}E_{1,-1}^{\mathrm{A}}=0,
\end{equation*}
and
\begin{equation}\label{Bell-fps-calculation:eq-A02}
  E_{0,+1}^{\mathrm{A}}+E_{0,-1}^{\mathrm{A}}=I_2\quad\text{ and }\quad
  E_{1,+1}^{\mathrm{A}}+E_{1,-1}^{\mathrm{A}}=I_2.
\end{equation}
The switching of
the measurement of the observable~$R$ or $Q$ in Step~A4,
depending on the outcome $c$ in Step~A3,
is realized by
a unitary time-evolution:
\begin{equation}\label{U4ToP3c=VcToP3c-Alice}
  U_\mathrm{A4}(\ket{\Theta}\otimes\ket{\Phi_\mathrm{A3}[c]})
  =(V_c^\mathrm{A}\ket{\Theta})\otimes\ket{\Phi_\mathrm{A3}[c]}
\end{equation}
for every $c\in\{0,1\}$
and every state $\ket{\Theta}$ of the composite system
consisting of the system $\mathcal{Q}_\mathrm{A}$ and the apparatus $\mathcal{A}_\mathrm{A4}$ explained below.
For each $c\in\{0,1\}$,
the operator $V_c^\mathrm{A}$ appearing in \eqref{U4ToP3c=VcToP3c-Alice} describes
a unitary time-evolution of the composite system consisting of
the system $\mathcal{Q}_\mathrm{A}$ and
an apparatus $\mathcal{A}_\mathrm{A4}$ measuring $\mathcal{Q}_\mathrm{A}$, and is
defined by the equation:
$$V_c^{\mathrm{A}}(\ket{\psi}\otimes\ket{\Phi_\mathrm{A4}^{\mathrm{init}}})
=\sum_{m=\pm 1}(E_{c,m}^{\mathrm{A}}\ket{\psi})\otimes\ket{\Phi_\mathrm{A4}[m]}$$
for every $\ket{\psi}\in\mathcal{H}_{\mathrm{A}}$.
The vector $\ket{\Phi_\mathrm{A4}^{\mathrm{init}}}$ is
the initial state of the apparatus $\mathcal{A}_\mathrm{A4}$,
and $\ket{\Phi_\mathrm{A4}[m]}$ is a final state of the apparatus $\mathcal{A}_\mathrm{A4}$
for each $m\in\{+1,-1\}$.
Thus,
the operator $V_c^\mathrm{A}$
describes the alternate measurement process of the qubit $\mathcal{Q}_\mathrm{A}$ sent from Charlie,
depending on the outcome $c$, on Alice's side.
Note that the unitarity of $U_\mathrm{A4}$ is confirmed by
Theorem~\ref{unitarity} below.
Actually, by setting
$N:=3$, $U_1:=V_0^\mathrm{A}$, $U_2:=V_1^\mathrm{A}$,
$P_1:=\product{\Phi_\mathrm{A3}[0]}{\Phi_\mathrm{A3}[0]}$,
$P_2:=\product{\Phi_\mathrm{A3}[1]}{\Phi_\mathrm{A3}[1]}$, and
$P_3:=I_\mathrm{A3}-P_1-P_2$ 
in Theorem~\ref{unitarity}, we have \eqref{U4ToP3c=VcToP3c-Alice},
where $I_\mathrm{A3}$ denotes the identity operator on the state space of
the apparatus $\mathcal{A}_\mathrm{A3}$.

\begin{theorem}\label{unitarity}
Let $\mathcal{H}$ and $\mathcal{K}$ be
two
arbitrary complex Hilbert spaces.
Let $U_1,\dots,U_N$ be $N$ arbitrary unitary operators on $\mathcal{H}$,
and let $\{P_n\}_{n=1}^N$ be an arbitrary PVM in $\mathcal{K}$.
Then
$$U:=U_1\otimes P_1 + \dots + U_N\otimes P_N$$
is a unitary operator on $\mathcal{H}\otimes \mathcal{K}$.
Moreover, for every $n=1,\dots,N$ and every $\ket{\Phi}\in\mathcal{K}$
if $P_n\ket{\Phi}=\ket{\Phi}$ then for every $\ket{\Theta}\in\mathcal{H}$ it holds that
$U(\ket{\Theta}\otimes\ket{\Phi})=(U_n\ket{\Theta})\otimes\ket{\Phi}$.
\end{theorem}

\begin{proof}
First, we see that
\begin{align*}
U^\dag U
&=\left(\sum_{n=1}^N U_n^\dag\otimes P_n\right)
\left(\sum_{l=1}^N U_l\otimes P_l\right)
=\sum_{n=1}^N \sum_{l=1}^N (U_n^\dag U_l)\otimes (P_n P_l) \\
&=\sum_{n=1}^N (U_n^\dag U_n)\otimes P_n
=I_{\mathcal{H}}\otimes\left(\sum_{n=1}^N P_n\right)=I_{\mathcal{H}}\otimes I_{\mathcal{K}} \\
&=I_{\mathcal{H}\otimes\mathcal{K}},
\end{align*}
where $I_{\mathcal{H}}$, $I_{\mathcal{K}}$, and $I_{\mathcal{H}\otimes\mathcal{K}}$ denote the identity operators on
$\mathcal{H}$, $\mathcal{K}$, and
$\mathcal{H}\otimes\mathcal{K}$, respectively.
Similarly, we can show that $UU^\dag =I$.
This completes the proof.
\end{proof}

\medskip

The sequential applications of $U_\mathrm{A3}$ and $U_\mathrm{A4}$ to the composite system
consisting of the two qubit system $\mathcal{Q}_\mathrm{A3}$ and $\mathcal{Q}_\mathrm{A}$ and
the apparatuses $\mathcal{A}_\mathrm{A3}$ and $\mathcal{A}_\mathrm{A4}$
result in the following single unitary time-evolution $U_{\mathrm{A}}$:
\begin{equation}\label{unitary-Alice-Bell}
U_{\mathrm{A}}(\ket{\Psi}\otimes\ket{\Phi_\mathrm{A3}^{\mathrm{init}}}\otimes\ket{\Phi_\mathrm{A4}^{\mathrm{init}}})
=\sum_{c=0,1}\sum_{m=\pm 1}((E_{c}\otimes E^{\mathrm{A}}_{c,m})
\ket{\Psi})\otimes\ket{\Phi_\mathrm{A3}[c]}\otimes\ket{\Phi_\mathrm{A4}[m]}
\end{equation}
for every $\ket{\Psi}\in\mathcal{H}_{\mathrm{A3}}\otimes\mathcal{H}_{\mathrm{A}}$,
where
\begin{equation}\label{eq:Ec=kcbc}
E_c:=\product{c}{c}.
\end{equation}

On the other hand,
each of Steps~B3 and B4 done by Bob is implemented by a unitary time-evolution
in a similar manner as follows:

\paragraph{Unitary implementation of Step~B3 done by Bob.}

To realize the coin tossing
by Bob
in Step~B3 of Protocol~\ref{Bell}
we make use of a measurement over a
single
qubit system.
Namely, to implement the Step~B3
we
introduce
a
single
qubit system $\mathcal{Q}_\mathrm{B3}$ with state space $\mathcal{H}_\mathrm{B3}$,
and perform a measurement over the system $\mathcal{Q}_\mathrm{B3}$ described
by a unitary time-evolution:
$$U_\mathrm{B3}(\ket{d}\otimes\ket{\Phi_\mathrm{B3}^{\mathrm{init}}})
=\ket{d}\otimes\ket{\Phi_\mathrm{B3}[d]}$$
for every $d\in\{0,1\}$,
where $\ket{d}\in\mathcal{H}_\mathrm{B3}$.
The vector $\ket{\Phi_\mathrm{B3}^{\mathrm{init}}}$ is the initial state of an apparatus
$\mathcal{A}_\mathrm{B3}$
measuring $\mathcal{Q}_\mathrm{B3}$,
and $\ket{\Phi_\mathrm{B3}[d]}$ is a final state of
the apparatus
$\mathcal{A}_\mathrm{B3}$
for each $d\in\{0,1\}$.
Prior to the measurement,
the system $\mathcal{Q}_\mathrm{B3}$ is prepared in the state $\ket{+}\in\mathcal{H}_\mathrm{B3}$.

\paragraph{Unitary implementation of Step~B4 done by Bob.}

Let $\{E^{\mathrm{B}}_{0,n}\}_{n=\pm 1}$ and $\{E^{\mathrm{B}}_{1,n}\}_{n=\pm 1}$ be
the PVMs in $\mathcal{H}_{\mathrm{B}}$ corresponding to the observables $S$ and $T$, respectively.
Namely, let
\begin{equation}\label{Bell-fps-calculation:eq-B01}
  S=E_{0,+1}^{\mathrm{B}}-E_{0,-1}^{\mathrm{B}}\quad\text{ and }\quad
  T=E_{1,+1}^{\mathrm{B}}-E_{1,-1}^{\mathrm{B}}
\end{equation}
be the spectral decompositions of $S$ and $T$, respectively,
where $\{E^{\mathrm{B}}_{0,n}\}_{n=\pm 1}$ and $\{E^{\mathrm{B}}_{1,n}\}_{n=\pm 1}$ are
collections of projectors
such that
\begin{equation*}
  E_{0,+1}^{\mathrm{B}}E_{0,-1}^{\mathrm{B}}=0\quad\text{ and }\quad
  E_{1,+1}^{\mathrm{B}}E_{1,-1}^{\mathrm{B}}=0,
\end{equation*}
and
\begin{equation}\label{Bell-fps-calculation:eq-B02}
  E_{0,+1}^{\mathrm{B}}+E_{0,-1}^{\mathrm{B}}=I_2\quad\text{ and }\quad
  E_{1,+1}^{\mathrm{B}}+E_{1,-1}^{\mathrm{B}}=I_2.
\end{equation}
The switching of
the measurement of the observable~$S$ or $T$ in Step~B4,
depending on the outcome $d$ in Step~B3,
is realized by
a unitary time-evolution:
\begin{equation}\label{U4ToP3c=VcToP3c-Bob}
  U_\mathrm{B4}(\ket{\Theta}\otimes\ket{\Phi_\mathrm{B3}[d]})
  =(V_d^\mathrm{B}\ket{\Theta})\otimes\ket{\Phi_\mathrm{B3}[d]}
\end{equation}
for every $d\in\{0,1\}$
and every state $\ket{\Theta}$ of the composite system
consisting of the system $\mathcal{Q}_\mathrm{B}$ and the apparatus $\mathcal{A}_\mathrm{B4}$ explained below.
For each $d\in\{0,1\}$,
the operator $V_d^\mathrm{B}$ appearing in \eqref{U4ToP3c=VcToP3c-Bob} describes
a unitary time-evolution of the composite system consisting of
the system $\mathcal{Q}_\mathrm{B}$ and
an apparatus $\mathcal{A}_\mathrm{B4}$ measuring $\mathcal{Q}_\mathrm{B}$, and is
defined by the equation:
$$V_d^{\mathrm{B}}(\ket{\psi}\otimes\ket{\Phi_\mathrm{B4}^{\mathrm{init}}})
=\sum_{n=\pm 1}(E_{d,n}^{\mathrm{B}}\ket{\psi})\otimes\ket{\Phi_\mathrm{B4}[n]}$$
for every $\ket{\psi}\in\mathcal{H}_{\mathrm{B}}$.
The vector $\ket{\Phi_\mathrm{B4}^{\mathrm{init}}}$ is
the initial state of the apparatus $\mathcal{A}_\mathrm{B4}$,
and $\ket{\Phi_\mathrm{B4}[n]}$ is a final state of the apparatus $\mathcal{A}_\mathrm{B4}$
for each $n\in\{+1,-1\}$.
Thus,
the operator $V_c^\mathrm{B}$
describes the alternate measurement process of the qubit $\mathcal{Q}_\mathrm{B}$ sent from Charlie,
depending on the outcome $d$, on Bob's side.
The unitarity of $U_\mathrm{B4}$ is confirmed by Theorem~\ref{unitarity}.

\bigskip

\smallskip

The sequential applications of $U_\mathrm{B3}$ and $U_\mathrm{B4}$ to the composite system
consisting of the two qubit system $\mathcal{Q}_\mathrm{B3}$ and $\mathcal{Q}_\mathrm{B}$ and
the apparatuses $\mathcal{A}_\mathrm{B3}$ and $\mathcal{A}_\mathrm{B4}$
result in the following single unitary time-evolution $U_{\mathrm{B}}$:
\begin{equation}\label{unitary-Bob-Bell}
U_{\mathrm{B}}(\ket{\Psi}\otimes\ket{\Phi_\mathrm{B3}^{\mathrm{init}}}\otimes\ket{\Phi_\mathrm{B4}^{\mathrm{init}}})
=\sum_{d=0,1}\sum_{n=\pm 1}((E_{d}\otimes E^{\mathrm{B}}_{d,n})
\ket{\Psi})\otimes\ket{\Phi_\mathrm{B3}[d]}\otimes\ket{\Phi_\mathrm{B4}[n]}
\end{equation}
for every $\ket{\Psi}\in\mathcal{H}_{\mathrm{B3}}\otimes\mathcal{H}_{\mathrm{B}}$,
where $E_d=\product{d}{d}$ as before.

Now, let us consider
a single
unitary time-evolution $U_{\mathrm{AB}}$
which describes
all the measurement processes over the composite system consisting of
$\mathcal{Q}_\mathrm{A3}$, $\mathcal{Q}_\mathrm{A}$, $\mathcal{Q}_\mathrm{B3}$,
and $\mathcal{Q}_\mathrm{B}$.
According to Postulates~\ref{composition} and \ref{evolution},
we have that
\begin{equation}\label{eq:UABTAotTB=UATAotUBTB}
  U_{\mathrm{AB}}(\ket{\Theta_{\mathrm{A}}}\otimes\ket{\Theta_{\mathrm{B}}})
  =(U_{\mathrm{A}}\ket{\Theta_{\mathrm{A}}})\otimes(U_{\mathrm{B}}\ket{\Theta_{\mathrm{B}}})
\end{equation}
for every state $\ket{\Theta_{\mathrm{A}}}$
of the composite system
consisting of the systems $\mathcal{Q}_\mathrm{A3}$ and  $\mathcal{Q}_\mathrm{A}$
and the apparatus $\mathcal{A}_\mathrm{A3}$ and $\mathcal{A}_\mathrm{A4}$
and every state $\ket{\Theta_{\mathrm{B}}}$
of the composite system
consisting of the systems $\mathcal{Q}_\mathrm{B3}$ and  $\mathcal{Q}_\mathrm{B}$
and the apparatus $\mathcal{A}_\mathrm{B3}$ and $\mathcal{A}_\mathrm{B4}$.
We use $\Omega$ to denote the alphabet $$\{0,1\}^2\times\{+1,-1\}^2.$$ Then,
for each $\ket{\Psi_{\mathrm{A}}}\in\mathcal{H}_{\mathrm{A3}}\otimes\mathcal{H}_{\mathrm{A}}$
and $\ket{\Psi_{\mathrm{B}}}\in\mathcal{H}_{\mathrm{B3}}\otimes\mathcal{H}_{\mathrm{B}}$,
using \eqref{unitary-Alice-Bell}, \eqref{unitary-Bob-Bell}, and \eqref{eq:UABTAotTB=UATAotUBTB} we see that
\begin{align*}
&U_{\mathrm{AB}}\left(\ket{\Psi_{\mathrm{A}}}\otimes\ket{\Psi_{\mathrm{B}}}\otimes\ket{\Phi^{\mathrm{init}}}\right) \\
&=\left(U_{\mathrm{A}}\left(\ket{\Psi_{\mathrm{A}}}\otimes
\ket{\Phi_{\mathrm{A3}}^{\mathrm{init}}}\otimes\ket{\Phi_{\mathrm{A4}}^{\mathrm{init}}}\right)\right)\otimes
\left(U_{\mathrm{B}}\left(\ket{\Psi_{\mathrm{B}}}\otimes
\ket{\Phi_{\mathrm{B3}}^{\mathrm{init}}}\otimes\ket{\Phi_{\mathrm{B4}}^{\mathrm{init}}}\right)\right) \\
&=
\sum_{(c,d,m,n)\in\Omega}
\left(\left(E_{c}\otimes E_{d}\otimes E^{\mathrm{A}}_{c,m}\otimes E^{\mathrm{B}}_{d,n}\right)
\left(\ket{\Psi_{\mathrm{A}}}\otimes\ket{\Psi_{\mathrm{B}}}\right)\right)\otimes\ket{\Phi[c,d,m,n]},
\end{align*}
where
$\ket{\Phi^{\mathrm{init}}}$ denotes
$\ket{\Phi_{\mathrm{A3}}^{\mathrm{init}}}\otimes\ket{\Phi_{\mathrm{B3}}^{\mathrm{init}}}\otimes
\ket{\Phi_{\mathrm{A4}}^{\mathrm{init}}}\otimes\ket{\Phi_{\mathrm{B4}}^{\mathrm{init}}}$,
and
$\ket{\Phi[c,d,m,n]}$ denotes
$\ket{\Phi_{\mathrm{A3}}[c]}\otimes\ket{\Phi_{\mathrm{B3}}[d]}\otimes\ket{\Phi_{\mathrm{A4}}[m]}\otimes\ket{\Phi_{\mathrm{B4}}[n]}$
for each $(c,d,m,n)\in\Omega$.
It follows from the linearity of $U_{\mathrm{AB}}$ that
\[
  U_{\mathrm{AB}}\left(\ket{\Psi}\otimes\ket{\Phi^{\mathrm{init}}}\right)
  =\sum_{(c,d,m,n)\in\Omega}
  \left(\left(E_{c}\otimes E_{d}\otimes E^{\mathrm{A}}_{c,m}\otimes E^{\mathrm{B}}_{d,n}\right)
  \ket{\Psi}\right)\otimes\ket{\Phi[c,d,m,n]}
\]
for every
$\ket{\Psi}\in
\mathcal{H}_{\mathrm{A3}}\otimes\mathcal{H}_{\mathrm{B3}}\otimes
\mathcal{H}_{\mathrm{A}}\otimes\mathcal{H}_{\mathrm{B}}$.
This $U_{\mathrm{AB}}$ describes the unitary time-evolution of Alice and Bob
in the repeated once of the infinite repetition of
that procedure in Protocol~\ref{Bell}
which consists of the four steps:
Steps~A3 and A4 on Alice's side and Steps~B3 and B4 on Bob's side. 
Totally, prior to the application of $U_{\mathrm{AB}}$, the total system consisting of
$\mathcal{Q}_\mathrm{A3}$, $\mathcal{Q}_\mathrm{B3}$,
$\mathcal{Q}_\mathrm{A}$ and $\mathcal{Q}_\mathrm{B}$
is prepared in the state
\begin{equation}\label{Bell-fps-calculation:eq04}
  \ket{\Psi^{\mathrm{init}}}:=\ket{+}\otimes\ket{+}\otimes\ket{\beta_{11}}.
\end{equation}

\subsection{Application of the principle of typicality}

The operator $U_{\mathrm{AB}}$ applying to the initial state $\ket{\Psi^{\mathrm{init}}}$ describes
the \emph{repeated once} of the infinite repetition of the measurements
in Protocol~\ref{Bell},
where the execution of
Steps~A3 and A4 and Steps~B3 and B4
is infinitely repeated.
Actually, we can
check that a collection
\begin{equation}\label{Bellmos}
  \left\{E_{c}\otimes E_{d}\otimes
  E^{\mathrm{A}}_{c,m}\otimes E^{\mathrm{B}}_{d,n}\right\}_{(c,d,m,n)\in\Omega}
\end{equation}
forms \emph{measurement operators},
as is confirmed by the following calculation:
\begin{align*}
  &\sum_{(c,d,m,n)\in\Omega}
  \left(E_{c}\otimes E_{d}\otimes E^{\mathrm{A}}_{c,m}\otimes E^{\mathrm{B}}_{d,n}\right)^\dag
  \left(E_{c}\otimes E_{d}\otimes E^{\mathrm{A}}_{c,m}\otimes E^{\mathrm{B}}_{d,n}\right) \\
  &=\sum_{(c,d,m,n)\in\Omega}
  E_{c}\otimes E_{d}\otimes E^{\mathrm{A}}_{c,m}\otimes E^{\mathrm{B}}_{d,n}
  =
  \sum_{c=0,1}\sum_{d=0,1}
  E_{c}\otimes E_{d}\otimes
  \left(\sum_{m=\pm 1}E^{\mathrm{A}}_{c,m}\right)\otimes
  \left(\sum_{n=\pm 1}E^{\mathrm{B}}_{d,n}\right) \\
  &=\left(\sum_{c=0,1}E_{c}\right)\otimes \left(\sum_{d=0,1}E_{d}\right)
  \otimes I_2\otimes I_2
  =I_2\otimes I_2\otimes I_2\otimes I_2 \\
  &=I_{16},
\end{align*}
where the third equality follows from \eqref{Bell-fps-calculation:eq-A02} and \eqref{Bell-fps-calculation:eq-B02},
and the forth equality follows from \eqref{eq:Ec=kcbc}.
Thus,
the total application $U_{\mathrm{AB}}$ of
$U_\mathrm{A3}$, $U_\mathrm{A4}$, $U_\mathrm{B3}$, and $U_\mathrm{B4}$
can be regarded as
a
\emph{single measurement} which is described by
the measurement operators \eqref{Bellmos}
and whose all possible outcomes form the set $\Omega$.

Hence, we can apply Definition~\ref{pmrpwst}
to this scenario
of
the setting of measurements.
Therefore, according to Definition~\ref{pmrpwst},
we can see that a \emph{world} is an infinite sequence over
$\Omega$
and the probability measure induced by
the
\emph{probability measure representation for the prefixes of
worlds}
is
a Bernoulli measure $\lambda_P$ on
$\Omega^\infty$,
where $P$ is a finite probability space on
$\Omega$
such that
$P(c,d,m,n)$ is the square of the norm of
the
vector
\begin{equation*}
  \left(\left(E_{c}\otimes E_{d}\otimes E^{\mathrm{A}}_{c,m}\otimes E^{\mathrm{B}}_{d,n}\right)
  \ket{\Psi^{\mathrm{init}}}\right)\otimes\ket{\Phi[c,d,m,n]}
\end{equation*}
for every $(c,d,m,n)\in\Omega$.
Here $\Omega$ equals the set of all possible records of the apparatus
in the \emph{repeated once} of the measurements.

Let us calculate the explicit form of $P(c,d,m,n)$.
First, using \eqref{Bell-fps-calculation:eq04}, \eqref{eq:ket-plus}, and \eqref{eq:Ec=kcbc},
we have that
\begin{equation}\label{Bell-fps-calculation:eq05}
  P(c,d,m,n)=\bra{+}E_{c}\ket{+}\bra{+}E_{d}\ket{+}
  \bra{\beta_{11}}\left(E^{\mathrm{A}}_{c,m}\otimes E^{\mathrm{B}}_{d,n}\right)\ket{\beta_{11}}
  =\frac{1}{4}
  \bra{\beta_{11}}\left(E^{\mathrm{A}}_{c,m}\otimes E^{\mathrm{B}}_{d,n}\right)\ket{\beta_{11}}
\end{equation}
for each $(c,d,m,n)\in\Omega$.
Then the inner-product
$\bra{\beta_{11}}(E^{\mathrm{A}}_{c,m}\otimes E^{\mathrm{B}}_{d,n})\ket{\beta_{11}}$
above is calculated as follows:
Using \eqref{Bell-fps-calculation:eq-A01} and \eqref{Bell-fps-calculation:eq-A02}
we have that
\begin{equation}\label{Bell-fps-calculation:eq-A03}
  E_{0,m}^{\mathrm{A}}=\frac{1}{2}(I_2+mR)\quad\text{ and }\quad
  E_{1,m}^{\mathrm{A}}=\frac{1}{2}(I_2+mQ)
\end{equation}
for every $m\in\{+1,-1\}$.
Similarly, using \eqref{Bell-fps-calculation:eq-B01} and \eqref{Bell-fps-calculation:eq-B02}
we have that
\begin{equation}\label{Bell-fps-calculation:eq-B03}
  E_{0,n}^{\mathrm{B}}=\frac{1}{2}(I_2+nS)\quad\text{ and }\quad
  E_{1,n}^{\mathrm{B}}=\frac{1}{2}(I_2+nT)
\end{equation}
for every $n\in\{+1,-1\}$.
On the other hand, simple calculations 
using \eqref{Bell-fps-calculation:eq01} and \eqref{Bell-fps-calculation:eq02} show that
\begin{align}
  \bra{\beta_{11}}(X\otimes I_2)\ket{\beta_{11}}
  &=\bra{\beta_{11}}(Z\otimes I_2)\ket{\beta_{11}}
  =\bra{\beta_{11}}(I_2\otimes X)\ket{\beta_{11}}
  =\bra{\beta_{11}}(I_2\otimes Z)\ket{\beta_{11}}
  =0, \label{Bell-fps-calculation:eq06a} \\
  \bra{\beta_{11}}(X\otimes Z)\ket{\beta_{11}}
  &=\bra{\beta_{11}}(Z\otimes X)\ket{\beta_{11}}
  =0, \label{Bell-fps-calculation:eq06b} \\
  \bra{\beta_{11}}(X\otimes X)\ket{\beta_{11}}
  &=\bra{\beta_{11}}(Z\otimes Z)\ket{\beta_{11}}
  =-1. \label{Bell-fps-calculation:eq06c}
\end{align}
Note from \eqref{Bell-fps-calculation:eq03} that
\begin{align*}
  \bra{\beta_{11}}(I_2\otimes S)\ket{\beta_{11}}
  &=-\frac{1}{\sqrt{2}}\bra{\beta_{11}}(I_2\otimes X)\ket{\beta_{11}}-\frac{1}{\sqrt{2}}\bra{\beta_{11}}(I_2\otimes Z)\ket{\beta_{11}}, \\
  \bra{\beta_{11}}(I_2\otimes T)\ket{\beta_{11}}
  &=-\frac{1}{\sqrt{2}}\bra{\beta_{11}}(I_2\otimes X)\ket{\beta_{11}}+\frac{1}{\sqrt{2}}\bra{\beta_{11}}(I_2\otimes Z)\ket{\beta_{11}}.
\end{align*}
Therefore, it follows from
\eqref{Bell-fps-calculation:eq03} and
\eqref{Bell-fps-calculation:eq06a} that
\begin{equation}\label{Bell-fps-calculation:eq07}
  \bra{\beta_{11}}(R\otimes I_2)\ket{\beta_{11}}
  =\bra{\beta_{11}}(Q\otimes I_2)\ket{\beta_{11}}
  =\bra{\beta_{11}}(I_2\otimes S)\ket{\beta_{11}}
  =\bra{\beta_{11}}(I_2\otimes T)\ket{\beta_{11}}
  =0.
\end{equation}
Note also from \eqref{Bell-fps-calculation:eq03} that
\begin{align*}
  \bra{\beta_{11}}(R\otimes S)\ket{\beta_{11}}
  &=-\frac{1}{\sqrt{2}}\bra{\beta_{11}}(X\otimes X)\ket{\beta_{11}}-\frac{1}{\sqrt{2}}\bra{\beta_{11}}(X\otimes Z)\ket{\beta_{11}}, \\
  \bra{\beta_{11}}(R\otimes T)\ket{\beta_{11}}
  &=-\frac{1}{\sqrt{2}}\bra{\beta_{11}}(X\otimes X)\ket{\beta_{11}}+\frac{1}{\sqrt{2}}\bra{\beta_{11}}(X\otimes Z)\ket{\beta_{11}}, \\
  \bra{\beta_{11}}(Q\otimes S)\ket{\beta_{11}}
  &=-\frac{1}{\sqrt{2}}\bra{\beta_{11}}(Z\otimes X)\ket{\beta_{11}}-\frac{1}{\sqrt{2}}\bra{\beta_{11}}(Z\otimes Z)\ket{\beta_{11}}, \\
  \bra{\beta_{11}}(Q\otimes T)\ket{\beta_{11}}
  &=-\frac{1}{\sqrt{2}}\bra{\beta_{11}}(Z\otimes X)\ket{\beta_{11}}+\frac{1}{\sqrt{2}}\bra{\beta_{11}}(Z\otimes Z)\ket{\beta_{11}}.
\end{align*}
Therefore,  it follows from \eqref{Bell-fps-calculation:eq06b} and \eqref{Bell-fps-calculation:eq06c} that
\begin{equation}\label{Bell-fps-calculation:eq08}
\begin{split}
  \bra{\beta_{11}}(R\otimes S)\ket{\beta_{11}}
  &=\bra{\beta_{11}}(R\otimes T)\ket{\beta_{11}}
  =\bra{\beta_{11}}(Q\otimes S)\ket{\beta_{11}}=\frac{1}{\sqrt{2}}, \\
  \bra{\beta_{11}}(Q\otimes T)\ket{\beta_{11}}&=-\frac{1}{\sqrt{2}}.
\end{split}
\end{equation}
Thus, based on \eqref{Bell-fps-calculation:eq-A03}, \eqref{Bell-fps-calculation:eq-B03},
\eqref{Bell-fps-calculation:eq07}, and \eqref{Bell-fps-calculation:eq08}, it is easy to see that
\begin{equation}\label{eq:b11EAcmotEBdnb11=1o41pfrac}
  \bra{\beta_{11}}\left(E^{\mathrm{A}}_{c,m}\otimes E^{\mathrm{B}}_{d,n}\right)\ket{\beta_{11}}
  =\frac{1}{4}\left[1+\frac{(-1)^{cd}mn}{\sqrt{2}}\right]
\end{equation}
for every $(c,d,m,n)\in\Omega$.
For example, in the case of $(c,d)=(1,1)$,
using \eqref{Bell-fps-calculation:eq-A03} and \eqref{Bell-fps-calculation:eq-B03}
we have that
\begin{align*}
  &\bra{\beta_{11}}\left(E^{\mathrm{A}}_{c,m}\otimes E^{\mathrm{B}}_{d,n}\right)\ket{\beta_{11}} \\
  &=\frac{1}{4}\left[\bra{\beta_{11}}(I_2\otimes I_2)\ket{\beta_{11}}
  +m\bra{\beta_{11}}(Q\otimes I_2)\ket{\beta_{11}}
  +n\bra{\beta_{11}}(I_2\otimes T)\ket{\beta_{11}}
  +mn\bra{\beta_{11}}(Q\otimes T)\ket{\beta_{11}}\right]
\end{align*}
for every $m,n\in\{+1,-1\}$.
Due to \eqref{Bell-fps-calculation:eq07} and \eqref{Bell-fps-calculation:eq08}
this results in the equation~\eqref{eq:b11EAcmotEBdnb11=1o41pfrac} with $(c,d)=(1,1)$,
as wanted.
Hence, using \eqref{Bell-fps-calculation:eq05} and \eqref{eq:b11EAcmotEBdnb11=1o41pfrac},
the explicit form of $P(c,d,m,n)$ is
obtained
as follows:
\begin{equation}\label{Pcdmn=1o161+-1cdmns2-Bell}
  P(c,d,m,n)=\frac{1}{16}\left[1+\frac{(-1)^{cd}mn}{\sqrt{2}}\right]%
\end{equation}
for every $(c,d,m,n)\in\Omega$.

Now, let us apply
Postulate~\ref{POT}, the \emph{principle of typicality},
to the setting of measurements
developed above.
Let $\alpha$ be \emph{our world} in the infinite repetition of the measurements in the above setting.
This $\alpha$
is an infinite sequence over $\Omega$
consisting
of records in the apparatuses
which is being generated by the infinite repetition of the measurements
described by the measurement operators~\eqref{Bellmos}
in the above setting.
Since the Bernoulli measure $\lambda_P$ on $\Omega^\infty$ is
the probability measure induced by the
probability measure representation
for the prefixes of
worlds
in the above setting,
it follows from
Postulate~\ref{POT} that
\emph{$\alpha$ is Martin-L\"of $P$-random}.

\subsection{Refined derivation of the equality among the conditional averages by quantum mechanics}

For each $c,d\in\{0,1\}$, we use $H(c,d)$ to denote the set $\{(c,d,m,n)\mid m=\pm 1\;\&\; n=\pm 1\}$.

Now,
let $c,d\in\{0,1\}$.
The set $H(c,d)$ consists of all records of the apparatuses
$\mathcal{A}_\mathrm{A3},\mathcal{A}_\mathrm{B3},\mathcal{A}_\mathrm{A4}$,
and $\mathcal{A}_\mathrm{B4}$,
in a repeated once of the procedure in Protocol~\ref{Bell}, where
\emph{Alice gets the outcome $c$ in Step~A3 and Bob gets the outcome $d$ in Step~B3}.
It follows from \eqref{Pcdmn=1o161+-1cdmns2-Bell}
that
$$P(H(c,d))=\frac{1}{4},$$
as expected from the point of view of the conventional quantum mechanics,
and
moreover
\begin{equation}\label{PBcdcdmn=Pcdmnsm7n7pm1Pcdm7n7-Bell}
  P_{H(c,d)}(c,d,m,n)=\frac{P(c,d,m,n)}{P(H(c,d))}
  =\frac{1}{4}\left[1+\frac{(-1)^{cd}mn}{\sqrt{2}}\right]
\end{equation}
for every $m,n\in\{+1,-1\}$.
Here we use the notation presented in Section~\ref{subsec:Conditional probability}.
Let $$\alpha_{c,d}:=\cond{H(c,d)}{\alpha}.$$
Then, since $\alpha$ is Martin-L\"of $P$-random,
using Theorem~\ref{conditional_probability} we have that
$\alpha_{c,d}$ is Martin-L\"of $P_{H(c,d)}$-random for the finite probability space $P_{H(c,d)}$ on $H(c,d)$.
Recall that $\cond{H(c,d)}{\alpha}$ is defined as an infinite sequence
over the alphabet $H(c,d)$
obtained from
$\alpha$ by eliminating all elements of $\Omega\setminus H(c,d)$ occurring in $\alpha$.
In other words, $\alpha_{c,d}$ is the subsequence of $\alpha$
consisting
only
of results that
Alice gets the outcome $c$ in Step~A3 and Bob gets the outcome $d$ in Step~B3.
For each $k\in\N^+$, we denote the $k$th element $\alpha_{c,d}(k)$ of the
subsequence
$\alpha_{c,d}$
as
\[
  \alpha_{c,d}(k)=(c,d,m_{c,d}(k),n_{c,d}(k)).
\]
In this manner, we introduce infinite sequences $m_{c,d}$ and $n_{c,d}$ over $\{+1,-1\}$.
The sequence $m_{c,d}$
is
\emph{the infinite sequence of outcomes of the
measurements
performed
by Alice in Step~A4
over the infinite repetition of the procedure in Protocol~\ref{Bell} in our world,
conditioned that Alice gets the outcome $c$ in Step~A3 and Bob gets the outcome $d$ in Step~B3}.
Similarly, the sequence $n_{c,d}$
is
\emph{the infinite sequence of outcomes of the
measurements
performed
by Bob in Step~B4
over the infinite repetition of the procedure in Protocol~\ref{Bell} in our world,
conditioned that Alice gets the outcome $c$ in Step~A3 and Bob gets the outcome $d$ in Step~B3}.

Based on the families $\{m_{c,d}\}$ and $\{n_{c,d}\}$ of infinite sequences over $\{+1,-1\}$,
the \emph{conditional averages}
$\langle RS\rangle$, $\langle QS\rangle$, $\langle RT\rangle$, and $\langle QT\rangle$
are
defined
as follows:
\begin{equation}\label{def-conditional-averages-QM}
\begin{split}
  &\langle RS\rangle:=\lim_{L\to\infty}\frac{1}{L}\sum_{k=1}^L m_{0,0}(k)n_{0,0}(k),
  \qquad\langle QS\rangle:=\lim_{L\to\infty}\frac{1}{L}\sum_{k=1}^L m_{1,0}(k)n_{1,0}(k), \\
  &\langle RT\rangle:=\lim_{L\to\infty}\frac{1}{L}\sum_{k=1}^L m_{0,1}(k)n_{0,1}(k),
  \qquad\langle QT\rangle:=\lim_{L\to\infty}\frac{1}{L}\sum_{k=1}^L m_{1,1}(k)n_{1,1}(k).
\end{split}
\end{equation}
The value $\langle RS\rangle$ can be interpreted as
\emph{the average value of the product of
outcome of the measurement of the observable $R$
performed
by Alice in Step~A4 and
outcome of the measurement of the observable $S$
performed
by Bob in Step~B4,
conditioned that Alice gets the outcome $0$ in Step~A3 and Bob gets the outcome $0$ in Step~B3}.
Note here that, in Protocol~\ref{Bell},
whenever Alice gets the outcome $0$
in Step~A3
she performs the measurement of $R$
over the first qubit in Step~A4, and
whenever Bob gets the outcome $0$
in Step~B3
he performs the measurement of $S$
over the second qubit in Step~B4.
An analogous interpretation can be made for each of
$\langle QS\rangle$, $\langle RT\rangle$, and $\langle QT\rangle$.

The conditional averages are calculated as follows:
Let $c,d\in\{0,1\}$.
Since the subsequence $\alpha_{c,d}$ is Martin-L\"{o}f $P_{H(c,d)}$-random,
it follows from Theorem~\ref{FI} and \eqref{PBcdcdmn=Pcdmnsm7n7pm1Pcdm7n7-Bell}
that for every $m,n=\pm 1$ it holds that
\begin{equation*}
  \lim_{L\to\infty}\frac{N_{(c,d,m,n)}(\rest{\alpha_{c,d}}{L})}{L}
  =\frac{1}{4}\left[1+\frac{(-1)^{cd}mn}{\sqrt{2}}\right].
\end{equation*}
Recall here that $N_{(c,d,m,n)}(\rest{\alpha_{c,d}}{L})$ denotes the number of the occurrences of $(c,d,m,n)$
in the prefix of $\alpha_{c,d}$ of length $L$.
Therefore, we have that
\begin{align*}
  &\lim_{L\to\infty}\frac{1}{L}\sum_{k=1}^L m_{c,d}(k)n_{c,d}(k) \\
  =&\lim_{L\to\infty}\frac{1}{L}\left[N_{(c,d,1,1)}(\rest{\alpha_{c,d}}{L})-N_{(c,d,1,-1)}(\rest{\alpha_{c,d}}{L})-N_{(c,d,-1,1)}(\rest{\alpha_{c,d}}{L})+N_{(c,d,-1,-1)}(\rest{\alpha_{c,d}}{L})\right] \\
  =&(-1)^{cd}\frac{1}{\sqrt{2}}.
\end{align*}
Thus,
the conditional averages
$\langle RS\rangle$, $\langle QS\rangle$, $\langle RT\rangle$, and $\langle QT\rangle$,
which are defined by \eqref{def-conditional-averages-QM},
are calculated as follows:
\[
  \langle RS\rangle=\frac{1}{\sqrt{2}},\quad
  \langle QS\rangle=\frac{1}{\sqrt{2}},\quad
  \langle RT\rangle=\frac{1}{\sqrt{2}},\quad
  \langle QT\rangle=-\frac{1}{\sqrt{2}}.
\]
Eventually,
we have the following equality for the conditional averages
as a result of the analysis of Protocol~\ref{Bell} in our
rigorous
framework
of quantum mechanics
based on the principle of typicality:
\begin{equation}\label{QM-equality}
  \langle RS\rangle + \langle QS\rangle + \langle RT\rangle - \langle QT\rangle =2\sqrt{2}.
\end{equation}
This equality has
exactly
the same form as
expected from
the aspect of
the conventional quantum mechanics,
i.e., as the equation~(2.230) in Section~2.6 of Nielsen and Chuang~\cite{NC00}.

\section{Refinement of the argument of local realism to lead to Bell's inequality}
\label{sec-Bell_inequality}

Nielsen and Chuang~\cite[Section~2.6]{NC00} describes
an analysis for Protocol~\ref{Bell} leading to Bell's inequality,
based on the \emph{assumptions of local realism}.
In this section,
we \emph{refine} and \emph{reformulate}
their
analysis and
the assumptions of local realism,
in the framework of
the
\emph{operational characterization of the notion of probability
by algorithmic randomness}, introduced and developed by our former works~\cite{T14,T15,T16arXiv}.
For that purpose, first we review
the framework of the operational characterization of the notion of probability
in the following subsection.

\subsection{Operational characterization of the notion of probability}
\label{subsec:OCNP}

The notion of probability plays an important role in almost all areas of science and technology.
In modern mathematics, however, probability theory means nothing other than measure theory,
and the operational characterization of the notion of probability is not established yet.
In our former works~\cite{T14,T15,T16arXiv,T19arXiv},
based on the toolkit of algorithmic randomness
we presented an operational characterization of the notion of probability
for \emph{discrete probability spaces},
including \emph{finite probability spaces}.
We used the notion of
\emph{Martin-L\"of $P$-randomness} in \cite{T14,T15,T16arXiv}
and that of \emph{its extension over Baire space} in \cite{T19arXiv}
to present the operational characterization.

According to Tadaki~\cite{T16arXiv}, 
in order to clarify our motivation and standpoint, and the meaning of the operational characterization,
let us consider a familiar example of a probabilistic phenomenon.
Specifically,
we
consider the repeated throwing of a fair die.
In this probabilistic phenomenon,
as throwing progressed,
a specific infinite sequence such as
\begin{equation*}
  3,5,6,3,4,2,2,3,6,1,5,3,5,4,1,\dotsc\dotsc\dotsc
\end{equation*}
is being generated,
where each number is the outcome of the corresponding throwing of the die.
Then the following
naive
question
may
arise naturally. 

\begin{quote}
\textbf{Question}:
What property should this infinite sequence
satisfy as a
probabilistic
phenomenon?
\end{quote}

Via a series of works~\cite{T14,T15,T16arXiv,T19arXiv},
we tried to answer this question.
We characterized the notion of probability
as
an infinite sequence of outcomes in a probabilistic phenomenon of a \emph{specific mathematical property}.
We called such an infinite sequence of outcomes
the \emph{operational characterization of the notion of probability}.
As the specific mathematical property,
we adopted
the notion of
\emph{Martin-L\"of $P$-randomness}.

In our former works~\cite{T14,T15,T16arXiv,T19arXiv},
we put forward this proposal as a thesis,
i.e., as Thesis~\ref{thesis} below.
We checked
the validity of
the thesis
based on
our intuitive understanding of the notion of probability.
Furthermore,
we characterized \emph{equivalently}
the basic notions in probability theory in terms of the operational characterization.
Namely, we equivalently characterized the notion of the \emph{independence} of random variables/events
in terms of
the operational characterization,
and we represented the notion of \emph{conditional probability}
in terms of
the operational characterization
in a natural way.
The existence of these equivalent characterizations confirms further the validity of the thesis.
See Tadaki~\cite{T16arXiv,T19arXiv}
for the detail of the operational characterization of the notion of probability
based on
Martin-L\"of $P$-randomness.

In this manner,
as revealed
by Tadaki~\cite{T14,T15,T16arXiv,T19arXiv},
the notion of Martin-L\"of $P$-randomness is thought to reflect all the properties of the notion of probability
based on
our intuitive understanding of the notion of probability.
Thus, \emph{we propose that a Martin-L\"of $P$-random sequence of elementary events
gives an operational characterization of the notion of probability},
as follows.

Let $\Omega$ be an alphabet, and let $P\in\PS(\Omega)$.
According to Tadaki~\cite{T16arXiv},
consider an infinite sequence $\alpha\in\Omega^\infty$ of outcomes which is
being generated
by infinitely repeated trials \emph{described by} the finite probability space $P$.
\emph{The operational characterization of the notion of probability
for the finite probability space $P$ is thought to be completed
if the property which the infinite sequence $\alpha$ has to satisfy is determined.}
\emph{We thus propose the following thesis}.

\begin{thesis}[Tadaki~\cite{T14,T15,T16arXiv}]\label{thesis}
Let $\Omega$ be an alphabet, and let $P\in\PS(\Omega)$.
An infinite sequence of outcomes in $\Omega$ which is
being generated by infinitely repeated trials \emph{described by}
the finite probability space $P$ on $\Omega$ is
a Martin-L\"of $P$-random sequence over $\Omega$.
\qed
\end{thesis}

Tadaki~\cite{T16arXiv}
demonstrated
the validity of Thesis~\ref{thesis}
from a variety of aspects.

\subsection{Refinement of the assumptions of local realism}

In what follows, we refine and reformulate
the analysis
for Protocol~\ref{Bell} to derive Bell's inequality,
given
in Section~2.6 of Nielsen and Chuang~\cite{NC00},
in the framework of the operational characterization of the notion of probability
reviewed in the preceding subsection.
Basically, we
follow the flow of the argument of
Nielsen and Chuang~\cite[Section~2.6]{NC00},
while refining it appropriately in terms of the operational characterization of the notion of probability.
Hence, we proceed according to Nielsen and Chuang~\cite[Section~2.6]{NC00}, as follows:

We first forget all the knowledge of quantum mechanics.
To obtain Bell's inequality, we analyze Protocol~\ref{Bell}
based on  ``our common sense notions of how the world works.''
Thus, we perform ``the common sense analysis'' for Protocol~\ref{Bell}.
In doing so, we are implicitly assuming the following two assumptions:
\begin{description}
\item[The assumption of realism:]
The assumption that the observables~$R$, $Q$, $S$, $T$ have
definite values~$r$, $q$, $s$, $t$, respectively, which exist independent of observation.
\item[The assumption of locality:]
The assumption that Alice performing her measurement does not influence
the result of Bob's measurement, and vice versa.
\end{description}
These two assumptions together are known as the \emph{assumptions of local realism}.
See Nielsen and Chuang~\cite[Section~2.6]{NC00}
for the considerations for the assumptions of local realism.

Now,
according to Nielsen and Chuang~\cite[Section~2.6]{NC00},
let us make ``the common sense analysis'' for Protocol~\ref{Bell},
based on the assumptions of local realism.
In Protocol~\ref{Bell},
Alice performs the measurement of either the observable $R$ or $Q$ over the first qubit in Step A4,
while Bob performs the measurement of either the observable $S$ or
$T$ over the second qubit in Step B4.
Based on the assumptions of local realism,
we
assume
that each of the observables $R$, $Q$, $S$, and $T$
has a specific value before the measurement, which is merely revealed by the measurement.
In particular,
\emph{in the terminology of the conventional probability theory},
we assume that $$p(r, q, s, t)$$ is the ``probability'' that,
before the measurements are performed,
the system is in a state where $R = r, Q =q, S = s$, and $T = t$.
This ``probability''
may depend on how Charlie performs his preparation,
and on experimental noise.
In the framework of the operational characterization of the notion of probability,
\emph{based on Thesis~\ref{thesis}},
the assumption above is refined and reformulated
in the following form:

\begin{assumption}\label{ASLR1}
Let $\omega$ be an infinite sequence of the values $(r,q,s,t)$ of the observables $R,Q,S,T$
which is being generated by
the infinite repetition of the procedure
in Protocol~\ref{Bell}.
Then there exists a finite probability space $H$ on $\{+1,-1\}^4$ such that
$\omega$ is a \emph{Martin-L\"of $H$-random} infinite sequence over $\{+1,-1\}^4$.
\qed
\end{assumption}

Assumption~\ref{ASLR1} is an \emph{operational refinement} of one of
the consequences
of the assumptions of local realism.
In Protocol~\ref{Bell},
Alice tosses a fair coin $C$ to get outcome $c\in\{0,1\}$ in Step A3,
while Bob tosses a fair coin $D$ to get outcome $d\in\{0,1\}$ in Step B3.
In the framework of the operational characterization of the notion of probability,
these probabilistic phenomena are refined and reformulated
in the following form:

\begin{assumption}\label{Impl-Thesis}
Let $\gamma$ be an infinite binary sequence which is
being generated by infinitely repeated tossing of the fair coin $C$ by Alice in Protocol~\ref{Bell}.
Then the infinite sequence $\gamma$ is a \emph{Martin-L\"of $U$-random} sequence over $\{0,1\}$,
where $U$ is a finite probability space on $\{0,1\}$ such that $U(0)=U(1)=1/2$.
Namely, $\gamma$ is \emph{Martin-L\"of random}.

Similarly,
let $\delta$ be an infinite binary sequence which is
being generated by infinitely repeated tossing of the fair coin $D$ by Bob in Protocol~\ref{Bell}.
Then the infinite sequence $\delta$ is also a \emph{Martin-L\"of $U$-random} sequence over $\{0,1\}$.
\qed
\end{assumption}

Assumption~\ref{Impl-Thesis} is just an implementation of Thesis~\ref{thesis}
in an infinite repetition of tossing of a fair coin.
In order to advance
``the common sense analysis'' for Protocol~\ref{Bell}
further
in a rigorous manner,
however,
we need to make
an additional
assumption for the relation
among the infinite sequences $\omega$, $\gamma$, and $\delta$.
Namely, the infinite sequences $\omega$, $\gamma$, and $\delta$ need to be \emph{independent}.
Thus, based on the notion of independence given in Definition~\ref{independency-of-ensembles},
we assume the following:
 
\begin{assumption}\label{AL2}
The infinite sequences $\omega$, $\gamma$, and $\delta$ are \emph{independent}.
\qed
\end{assumption}

Assumption~\ref{AL2} is an \emph{operational refinement} of one of
the consequences
of the assumption of locality.

Now, according to Definition~\ref{independency-of-ensembles},
Assumption~\ref{ASLR1}, Assumption~\ref{Impl-Thesis}, and Assumption~\ref{AL2}
together
imply
the following \emph{single} assumption:

\begin{assumption}\label{ALRsingle}
Let $\omega$ be an infinite sequence of the values $(r,q,s,t)$ of the observables $R,Q,S,T$
which is being generated by
the infinite repetition of the procedure
in Protocol~\ref{Bell}.
Let $\gamma$ be an infinite binary sequence which is
being generated by infinitely repeated tossing of the fair coin $C$ by Alice in Protocol~\ref{Bell},
and let $\delta$ be an infinite binary sequence which is
being generated by infinitely repeated tossing of the fair coin $D$ by Bob in Protocol~\ref{Bell}.
Then there exists a finite probability space $H$ on $\{+1,-1\}^4$ such that
the infinite sequence $\omega\times\gamma\times\delta$ over $\{+1,-1\}^4\times\{0,1\}\times\{0,1\}$
is \emph{Martin-L\"of $H\times U\times U$-random},
where $U$ is a finite probability space on $\{0,1\}$ such that $U(0)=U(1)=1/2$.
\qed
\end{assumption}

However, it follows from Theorem~\ref{independency-imlplies-each-randomness} that
if the infinite sequence $\omega\times\gamma\times\delta$ over $\{+1,-1\}^4\times\{0,1\}\times\{0,1\}$
is Martin-L\"of $H\times U\times U$-random then
the infinite sequence $\omega$ over $\{+1,-1\}^4$ is Martin-L\"of $H$-random and
the infinite binary sequences $\gamma$ and $\delta$ are Martin-L\"of $U$-random.
Thus, in fact, Assumption~\ref{ALRsingle} is \emph{equivalent} to
the conjunction of
Assumption~\ref{ASLR1}, Assumption~\ref{Impl-Thesis}, and Assumption~\ref{AL2}.
Hence, Assumption~\ref{ALRsingle} \emph{alone} serves
as an \emph{operational refinement} of the whole of the assumptions of local realism in the context of Protocol~\ref{Bell}.

\begin{remark}\label{rem:RC-Bell}
In the context of the \emph{relativized computation},
we can consider the notion of \emph{Martin-L\"of $P$-randomness relative to infinite sequences}.
Theorem~42 of Tadaki~\cite{T16arXiv} states that
the notion of independence of Martin-L\"of $P$-random infinite sequences can be 
\emph{equivalently}
represented by the notion of Martin-L\"of $P$-randomness relative to infinite sequences.
Thus,
based on Theorem~42 of Tadaki~\cite{T16arXiv},
in Assumption~\ref{ALRsingle}
the statement
``the infinite sequence $\omega\times\gamma\times\delta$ over $\{+1,-1\}^4\times\{0,1\}\times\{0,1\}$
is \emph{Martin-L\"of $H\times U\times U$-random}''
can be \emph{equivalently} rephrased
as the statement that
\begin{enumerate}
  \item the infinite sequence $\omega$ is Martin-L\"of $H$-random,
  \item the infinite sequence $\gamma$ is Martin-L\"of random
    \emph{relative to $\omega$}, and
  \item the infinite sequence $\delta$ is Martin-L\"of random
    \emph{relative to $\gamma$ and $\omega$}.
\end{enumerate}
Note that in establishing this equivalence
we do not
have
to impose any computability restrictions on $H$ at all:
In Assumption~\ref{ALRsingle} (or in Assumption~\ref{ASLR1})
the finite probability space $H$ on $\{+1,-1\}^4$ can be chosen completely arbitrarily
without failing this equivalence (see Theorem~42 of Tadaki~\cite{T16arXiv}).
\qed
\end{remark}

\subsection{Refined derivation of Bell's inequality among the conditional averages}

Based on Assumption~\ref{ALRsingle}
(or equivalently,
the conjunction of Assumptions~\ref{ASLR1}, \ref{Impl-Thesis}, and \ref{AL2}),
let us derive Bell's inequality
in the framework of the operational characterization of the notion of probability.
Let $$\alpha:=\omega\times\gamma\times\delta.$$
Then, $\alpha$ is Martin-L\"of $H\times U\times U$-random by Assumption~\ref{ALRsingle}.
On the other hand,
note that
\begin{equation}\label{PtUtUrqst=f14rqst-Local-Realism}
  (H\times U\times U)((r,q,s,t),c,d)=\frac{1}{4}H(r,q,s,t)
\end{equation}
for every $r,q,s,t\in\{+1,-1\}$ and every $c,d\in\{0,1\}$.
Here we use the notation presented in Section~\ref{sec-IMLP}.
For each $c,d\in\{0,1\}$, we use $G(c,d)$ to denote the set
$$\left\{(x,c,d)\,\middle\vert\, x\in\{+1,-1\}^4\right\}.$$

Now,
let $c,d\in\{0,1\}$.
The set $G(c,d)$ consists of all possible
results
in a repeated once of the procedure in Protocol~\ref{Bell}, where
\emph{Alice gets the outcome $c$ in Step~A3 and Bob gets the outcome $d$ in Step~B3}.
\emph{In the terminology of the conventional probability theory},
$(H\times U\times U)(G(c,d))$ is
the ``probability'' that
\emph{Alice gets the outcome $c$ in Step~A3 and Bob gets the outcome $d$ in Step~B3}.
Actually, it follows from \eqref{PtUtUrqst=f14rqst-Local-Realism} that
$$(H\times U\times U)(G(c,d))=\frac{1}{4},$$
as expected from the conventional probability theory.
Thus, it
follows from \eqref{PtUtUrqst=f14rqst-Local-Realism} again that
\begin{equation}\label{PUUGcdrqstcd=prqst-Local-Realism}
  (H\times U\times U)_{G(c,d)}((r,q,s,t),c,d)
  =\frac{(H\times U\times U)((r,q,s,t),c,d)}{(H\times U\times U)(G(c,d))}
  =H(r,q,s,t)
\end{equation}
for every $r,q,s,t\in\{+1,-1\}$.
Here we use the notation presented in Section~\ref{subsec:Conditional probability}.
Let $$\alpha_{c,d}:=\cond{G(c,d)}{\alpha}.$$
Then, since $\alpha$ is Martin-L\"of $H\times U\times U$-random,
using Theorem~\ref{conditional_probability} we have that
$\alpha_{c,d}$ is Martin-L\"of $(H\times U\times U)_{G(c,d)}$-random
for the finite probability space $(H\times U\times U)_{G(c,d)}$ on $G(c,d)$.
Recall that $\cond{G(c,d)}{\alpha}$ is defined as an infinite sequence
over the alphabet $G(c,d)$
obtained from
$\alpha$ by eliminating all elements of $(\{+1,-1\}^4\times\{0,1\}\times\{0,1\})\setminus G(c,d)$
occurring in $\alpha$.
In other words, $\alpha_{c,d}$ is the subsequence of $\alpha$
consisting of results that
Alice gets the outcome $c$ in Step~A3 and Bob gets the outcome $d$ in Step~B3.
For each $k\in\N^+$, we denote the $k$th element $\alpha_{c,d}(k)$ of the
subsequence
$\alpha_{c,d}$
as
\[
  \alpha_{c,d}(k)=((r_{c,d}(k),q_{c,d}(k),s_{c,d}(k),t_{c,d}(k)),c,d).
\]
In this manner, we introduce infinite sequences $r_{c,d}$, $q_{c,d}$, $s_{c,d}$, and $t_{c,d}$ over $\{+1,-1\}$.
The sequence $r_{c,d}$
is
\emph{the infinite sequence of
values of the observable $R$
over the infinite repetition of the procedure in Protocol~\ref{Bell},
conditioned that Alice gets the outcome $c$ at Step~A3 and Bob gets the outcome $d$ at Step~B3}.
The sequences $q_{c,d}$, $s_{c,d}$, and $t_{c,d}$ have
an analogous meaning with respect to the observables $Q$, $S$, and $T$, respectively.

Based on the families $\{r_{c,d}\}$, $\{q_{c,d}\}$, $\{s_{c,d}\}$, and $\{t_{c,d}\}$ of infinite sequences over $\{+1,-1\}$,
the \emph{conditional averages}
$\langle RS\rangle$, $\langle QS\rangle$, $\langle RT\rangle$, and $\langle QT\rangle$ are defined
as follows:
\begin{equation}\label{def-conditional-averages-ALR}
\begin{split}
  &\langle RS\rangle:=\lim_{L\to\infty}\frac{1}{L}\sum_{k=1}^L r_{0,0}(k)s_{0,0}(k),
  \qquad\langle QS\rangle:=\lim_{L\to\infty}\frac{1}{L}\sum_{k=1}^L q_{1,0}(k)s_{1,0}(k), \\
  &\langle RT\rangle:=\lim_{L\to\infty}\frac{1}{L}\sum_{k=1}^L r_{0,1}(k)t_{0,1}(k),
  \qquad\langle QT\rangle:=\lim_{L\to\infty}\frac{1}{L}\sum_{k=1}^L q_{1,1}(k)t_{1,1}(k).
\end{split}
\end{equation}
The value $\langle RS\rangle$ can be interpreted as
\emph{the average value of the product of
outcome of the measurement of the observable $R$
performed
by Alice in Step~A4 and
outcome of the measurement of the observable $S$
performed
by Bob in Step~B4,
conditioned that Alice gets the outcome $0$ in Step~A3 and Bob gets the outcome $0$ in Step~B3}.
Note here that, in Protocol~\ref{Bell},
whenever Alice gets the outcome $0$
in Step~A3
she performs the measurement of $R$
over the first qubit in Step~A4, and
whenever Bob gets the outcome $0$
in Step~B3
he performs the measurement of $S$
over the second qubit in Step~B4.
An analogous
interpretation can be made for each of
$\langle QS\rangle$, $\langle RT\rangle$, and $\langle QT\rangle$.

The conditional averages are calculated as follows:
Let $c,d\in\{0,1\}$.
Since the subsequence $\alpha_{c,d}$ is Martin-L\"{o}f $(H\times U\times U)_{G(c,d)}$-random,
it follows from Theorem~\ref{FI} and \eqref{PUUGcdrqstcd=prqst-Local-Realism}
that for every $r,q,s,t=\pm 1$ it holds that
\begin{equation*}
  \lim_{L\to\infty}\frac{N_{((r,q,s,t),c,d)}(\rest{\alpha_{c,d}}{L})}{L}=H(r,q,s,t).
\end{equation*}
Recall here that
$N_{((r,q,s,t),c,d)}(\rest{\alpha_{c,d}}{L})$ denotes the number of the occurrences of
$((r,q,s,t),c,d)$
in the prefix of $\alpha_{c,d}$ of length $L$.
Thus,
the conditional averages
$\langle RS\rangle$, $\langle QS\rangle$, $\langle RT\rangle$, and $\langle QT\rangle$
are calculated as follows:
\begin{equation}\label{EF-Conditional-Averages-Local-Realism}
\begin{split}
  \langle RS\rangle
  =&\lim_{L\to\infty}\frac{1}{L}\sum_{r,q,s,t=\pm 1} N_{((r,q,s,t),0,0)}(\rest{\alpha_{0,0}}{L})rs
  =\sum_{r,q,s,t=\pm 1} H(r,q,s,t)rs, \\
  \langle QS\rangle
  =&\lim_{L\to\infty}\frac{1}{L}\sum_{r,q,s,t=\pm 1} N_{((r,q,s,t),1,0)}(\rest{\alpha_{1,0}}{L})qs
  =\sum_{r,q,s,t=\pm 1} H(r,q,s,t)qs, \\
  \langle RT\rangle
  =&\lim_{L\to\infty}\frac{1}{L}\sum_{r,q,s,t=\pm 1} N_{((r,q,s,t),0,1)}(\rest{\alpha_{0,1}}{L})rt
  =\sum_{r,q,s,t=\pm 1} H(r,q,s,t)rt, \\
  \langle QT\rangle
  =&\lim_{L\to\infty}\frac{1}{L}\sum_{r,q,s,t=\pm 1} N_{((r,q,s,t),1,1)}(\rest{\alpha_{1,1}}{L})qt
  =\sum_{r,q,s,t=\pm 1} H(r,q,s,t)qt,
\end{split}
\end{equation}
where the sums are over all $(r,q,s,t)\in\{+1,-1\}^4$.

Let $(r,q,s,t)\in\{+1,-1\}^4$.
Note that $rs + qs + rt - qt = (r + q)s + (r - q)t$.
Then, since $q,r = \pm 1$,
we have that
either $(r + q)s = 0$ or $(r - q)t = 0$.
It follows that
$rs + qs + rt - qt = \pm 2$.
Thus, using \eqref{EF-Conditional-Averages-Local-Realism} we see that
\begin{align*}
  &\langle RS\rangle + \langle QS\rangle + \langle RT\rangle - \langle QT\rangle \\
  &=\sum_{r,q,s,t=\pm 1} H(r,q,s,t)rs+\sum_{r,q,s,t=\pm 1} H(r,q,s,t)qs
    +\sum_{r,q,s,t=\pm 1} H(r,q,s,t)rt\\
  &\hspace*{83mm}-\sum_{r,q,s,t=\pm 1} H(r,q,s,t)qt \\
  &=\sum_{r,q,s,t=\pm 1} H(r,q,s,t)(rs+qs+rt-qt) \\
  &\le \sum_{r,q,s,t=\pm 1} H(r,q,s,t)\,\times 2 \\
  &=2.
\end{align*}
Hence, we finally obtain
the \emph{Bell's inequality} (also known as the \emph{CHSH inequality}),
\begin{equation}\label{Bell-inequality}
   \langle RS\rangle + \langle QS\rangle + \langle RT\rangle - \langle QT\rangle\le 2.
\end{equation}
This inequality has the same form as expected
from the conventional probability theory,
i.e., the same form as one obtained
by means of ``the common sense analysis'' based on the assumptions of local realism,
which is performed in Nielsen and Chuang~\cite[Section~2.6]{NC00} in a \emph{vague} manner.

\section{Refinement of the analysis of the GHZ experiment by quantum mechanics,
based on the principle of typicality}
\label{sec:GHZ-QM}

In this section,
based on the \emph{principle of typicality},
we \emph{refine} and \emph{reformulate}
the
analysis for the protocol of the GHZ experiment by
quantum mechanics,
which is presented by Mermin~\cite{Mer90}
based on the result of Greenberger, Horne, and Zeilinger~\cite{GHZ89}.
Thus, in what follows,
we investigate
Protocol~\ref{GHZ} below
due to Greenberger, Horne, and Zeilinger~\cite{GHZ89}, and Mermin~\cite{Mer90}
in \emph{our
rigorous
framework of MWI based on the principle of typicality},
developed in Section~\ref{MWI}.

Consider a system of three spin-$1/2$ particles, each numbered $i=1,2,3$,
and consider their observables
$\displaystyle
A^i_0:=\sigma^i_x$
and
$\displaystyle
A^i_1:=\sigma^i_y$
for each particle $i=1,2,3$,
where $\sigma_x=X$ and $\sigma_y=Y$.
The
\emph{Greenberger-Horne-Zeilinger state} (\emph{GHZ state}, for short) \cite{GHSZ90}
is a state of 
a system of three spin-$1/2$ particles defined by
\begin{equation}\label{GHZ-fps-calculation:eq01}
  \ket{\text{GHZ}}
  :=\frac{\ket{+1}\otimes\ket{+1}\otimes\ket{+1}-\ket{-1}\otimes\ket{-1}\otimes\ket{-1}}{\sqrt{2}}.
\end{equation}
Here we denote $\ket{0}$ and $\ket{1}$ by $\ket{+1}$ and $\ket{-1}$, respectively,
and therefore
$Z\ket{+1}=\ket{+1}$ and $Z\ket{-1}=-\ket{-1}$
hold
due to \eqref{Bell-fps-calculation:eq02}.

\begin{protocol}\label{GHZ}
The protocol involves four parties,
$\mathrm{P}_0, \mathrm{P}_1, \mathrm{P}_2, \mathrm{P}_3$.
They together repeat the following procedure \emph{forever}.
{
\renewcommand{\labelenumi}{Step \arabic{enumi}:}
\setlength{\leftmargini}{40pt}
\begin{enumerate}
\item The party~$\mathrm{P}_0$ prepares a system of three spin-$1/2$ particles
in the
state~$\ket{\text{GHZ}}$.
\item The party~$\mathrm{P}_0$ passes the three particles $1, 2, 3$ to
the
three parties $\mathrm{P}_1, \mathrm{P}_2, \mathrm{P}_3$, respectively.
\end{enumerate}
}

Then each party $\mathrm{P}_i$ performs the following $(i=1,2,3)$:
{
\renewcommand{\labelenumi}{Step $\mathrm{P}_i$\arabic{enumi}:}
\setlength{\leftmargini}{51pt}
\begin{enumerate}
\setcounter{enumi}{2}
\item The party $\mathrm{P}_i$ tosses a fair coin $C_i$ to get outcome $c_i\in\{0,1\}$.
\item The party $\mathrm{P}_i$ performs the measurement of either $A^i_0$ or $A^i_1$
  over the particle $i$ to obtain outcome $m_i\in\{+1,-1\}$,
  depending on $c_i=0$ or $1$.
\end{enumerate}
\renewcommand{\labelenumi}{(\roman{enumi})}
}

While repeating the above procedure forever,
the
parties $\mathrm{P}_1, \mathrm{P}_2, \mathrm{P}_3$
(occasionally)
compare their measurement outcomes with one another,
in order to find the breakdown of
the \emph{perfect correlations over the three parties}
stated in Theorem~\ref{thm:GHZ} below.
\qed
\end{protocol}

In what follows,
we analyze Protocol~\ref{GHZ}
in our framework of quantum mechanics based on
Postulate~\ref{POT}, the principle of typicality,
together with Postulates~\ref{state_space}, \ref{composition}, and \ref{evolution}.
To complete this,
\emph{we have to implement everything in
Steps~$\mathrm{P}_i$3 and $\mathrm{P}_i$4 of Protocol~\ref{GHZ}
for all $i=1,2,3$
by unitary time-evolution}.

\subsection{\boldmath Unitary implementation of all steps done by the
three
parties $\mathrm{P}_1, \mathrm{P}_2, \mathrm{P}_3$}

For each $i=1,2,3$,
we denote the system of the qubit which the party $\mathrm{P}_0$ sends to the party $\mathrm{P}_i$ in Step~2
by $\mathcal{Q}_{\mathrm{P}_i}$ with state space $\mathcal{H}_{\mathrm{P}_i}$.
We analyze
Steps~$\mathrm{P}_i$3 and $\mathrm{P}_i$4 with all $i=1,2,3$
of Protocol~\ref{GHZ}
in our framework of quantum mechanics based on the principle of typicality.
In particular, we realize each of the coin tossings
in Steps~$\mathrm{P}_13$, $\mathrm{P}_23$, $\mathrm{P}_33$
by a measurement of the observable $\ket{1}\bra{1}$ over
a system of a single qubit in the state $\ket{+}$,
which is defined by \eqref{eq:ket-plus}.
We will then describe all the measurement processes during
Steps~$\mathrm{P}_i$3 and $\mathrm{P}_i$4 with all $i=1,2,3$
as a single unitary interaction between systems and apparatuses.

For each $i=1,2,3$, each of Steps~$\mathrm{P}_i$3 and $\mathrm{P}_i$4 done
by the party $\mathrm{P}_i$ is implemented by a unitary time-evolution
in the following manner:

\paragraph{\boldmath Unitary implementation of Step~$\mathrm{P}_i$3 done by the party $\mathrm{P}_i$.}

To realize the coin tossing by the party~$\mathrm{P}_i$
in Step~$\mathrm{P}_i$3 of Protocol~\ref{GHZ}
we make use of a measurement over a single qubit system.
Namely, to implement
Step~$\mathrm{P}_i$3
we introduce a single qubit system $\mathcal{Q}_{\mathrm{P}_i3}$
with state space $\mathcal{H}_{\mathrm{P}_i3}$,
and perform a measurement over the system $\mathcal{Q}_{\mathrm{P}_i3}$ described
by a unitary time-evolution:
$$U_{\mathrm{P}_i3}(\ket{c}\otimes\ket{\Phi_{\mathrm{P}_i3}^{\mathrm{init}}})
=\ket{c}\otimes\ket{\Phi_{\mathrm{P}_i3}[c]}$$
for every $c\in\{0,1\}$,
where $\ket{c}\in\mathcal{H}_{\mathrm{P}_i3}$.
The vector $\ket{\Phi_{\mathrm{P}_i3}^{\mathrm{init}}}$ is
the initial state of an apparatus $\mathcal{A}_{\mathrm{P}_i3}$
measuring $\mathcal{Q}_{\mathrm{P}_i3}$,
and $\ket{\Phi_{\mathrm{P}_i3}[c]}$ is a final state of
the apparatus $\mathcal{A}_{\mathrm{P}_i3}$ for each $c\in\{0,1\}$.%
\footnote{We assume, of course, the orthogonality of the final states $\ket{\Phi_{\mathrm{P}_i3}[c]}$, i.e.,
the property that
$\braket{\Phi_{\mathrm{P}_i3}[c]}{\Phi_{\mathrm{P}_i3}[c']}=\delta_{c,c'}$.
Furthermore,
we assume the orthogonality of the finial states for
each of all apparatuses which appear in the rest of this section.}
Prior to the measurement,
the system $\mathcal{Q}_{\mathrm{P}_i3}$ is prepared in the state $\ket{+}\in\mathcal{H}_{\mathrm{P}_i3}$.

\paragraph{\boldmath Unitary implementation of Step~$\mathrm{P}_i$4 done by the party $\mathrm{P}_i$.}

Let $\{E^i_{0,m}\}_{m=\pm 1}$ and $\{E^i_{1,m}\}_{m=\pm 1}$ be the PVMs in $\mathcal{H}_{\mathrm{P}_i}$
corresponding to the observables $A^i_0$ and $A^i_1$, respectively.
Namely, let
\begin{equation}\label{GHZ-fps-calculation:eq02}
  A^i_0=E^i_{0,+1}-E^i_{0,-1}\quad\text{ and }\quad
  A^i_1=E^i_{1,+1}-E^i_{1,-1}.
\end{equation}
be the spectral decompositions of $A^i_0$ and $A^i_1$, respectively,
where $\{E^i_{0,m}\}_{m=\pm 1}$ and $\{E^i_{1,m}\}_{m=\pm 1}$
are collections of projectors such that
\begin{equation*}
  E_{0,+1}^i E_{0,-1}^i=0\quad\text{ and }\quad
  E_{1,+1}^i E_{1,-1}^i=0,
\end{equation*}
and
\begin{equation}\label{GHZ-fps-calculation:eq03}
  E_{0,+1}^i+E_{0,-1}^i=I_2\quad\text{ and }\quad
  E_{1,+1}^i+E_{1,-1}^i=I_2.
\end{equation}
The switching of the measurement of the observable~$A^i_0$ or $A^i_1$ in Step~$\mathrm{P}_i$4,
depending on the outcome $c_i$ in Step~$\mathrm{P}_i$3,
is realized by
a unitary time-evolution:
\begin{equation}\label{U4ToPi3c=VcToPi3c-Pi}
  U_{\mathrm{P}_i4}(\ket{\Theta}\otimes\ket{\Phi_{\mathrm{P}_i3}[c]})
  =(V_c^{\mathrm{P}_i}\ket{\Theta})\otimes\ket{\Phi_{\mathrm{P}_i3}[c]}
\end{equation}
for every $c\in\{0,1\}$
and every state $\ket{\Theta}$ of the composite system consisting of
the system $\mathcal{Q}_{\mathrm{P}_i}$ and the apparatus $\mathcal{A}_{\mathrm{P}_i4}$
explained below.
For each $c\in\{0,1\}$,
the operator $V_c^{\mathrm{P}_i}$ appearing in \eqref{U4ToPi3c=VcToPi3c-Pi}
describes a unitary time-evolution of the composite system consisting of
the system $\mathcal{Q}_{\mathrm{P}_i}$ and
an apparatus $\mathcal{A}_{\mathrm{P}_i4}$ measuring $\mathcal{Q}_{\mathrm{P}_i}$,
and is defined by the equation:
$$V_c^{\mathrm{P}_i}(\ket{\psi}\otimes\ket{\Phi_{\mathrm{P}_i4}^{\mathrm{init}}})
=\sum_{m=\pm 1}(E_{c,m}^i\ket{\psi})\otimes\ket{\Phi_{\mathrm{P}_i4}[m]}$$
for every $\ket{\psi}\in\mathcal{H}_{\mathrm{P}_i}$.
The vector $\ket{\Phi_{\mathrm{P}_i4}^{\mathrm{init}}}$ is
the initial state of the apparatus $\mathcal{A}_{\mathrm{P}_i4}$,
and $\ket{\Phi_{\mathrm{P}_i4}[m]}$ is a final state of the apparatus $\mathcal{A}_{\mathrm{P}_i4}$
for each $m\in\{+1,-1\}$.
Thus,
the operator $V_c^{\mathrm{P}_i}$
describes the alternate measurement process of the qubit $\mathcal{Q}_{\mathrm{P}_i}$
sent from the party $\mathrm{P}_0$, depending on the outcome $c_i$, on the party $\mathrm{P}_i$'s side.
Note that the unitarity of $U_{\mathrm{P}_i4}$ is confirmed by Theorem~\ref{unitarity}.

\bigskip

\medskip

The sequential applications of $U_{\mathrm{P}_i3}$ and $U_{\mathrm{P}_i4}$
to the composite system consisting of
the two qubit system $\mathcal{Q}_{\mathrm{P}_i3}$ and $\mathcal{Q}_{\mathrm{P}_i}$ and
the apparatuses $\mathcal{A}_{\mathrm{P}_i3}$ and $\mathcal{A}_{\mathrm{P}_i4}$
result in the following single unitary time-evolution $U_{\mathrm{P}_i}$:
\begin{equation}\label{unitary-Pi-GHZ}
U_{\mathrm{P}_i}(\ket{\Psi}\otimes\ket{\Phi_{\mathrm{P}_i3}^{\mathrm{init}}}\otimes\ket{\Phi_{\mathrm{P}_i4}^{\mathrm{init}}})
=\sum_{c=0,1}\sum_{m=\pm 1}((E_{c}\otimes E^i_{c,m})
\ket{\Psi})\otimes\ket{\Phi_{\mathrm{P}_i3}[c]}\otimes\ket{\Phi_{\mathrm{P}_i4}[m]}
\end{equation}
for every $\ket{\Psi}\in\mathcal{H}_{\mathrm{P}_i3}\otimes\mathcal{H}_{\mathrm{P}_i}$,
where the projector~$E_c$ on $\mathcal{H}_{\mathrm{P}_i3}$ is defined by \eqref{eq:Ec=kcbc}.

Now, let us consider 
a single unitary time-evolution $U_{\mathrm{P}_1,\mathrm{P}_2,\mathrm{P}_3}$
which describes all the measurement processes over the composite system consisting of
$\mathcal{Q}_{\mathrm{P}_i3}$ and $\mathcal{Q}_{\mathrm{P}_i}$ with all $i=1,2,3$.
According to Postulates~\ref{composition} and \ref{evolution},
we have that
\begin{equation}\label{eq:UP1P2P3T1otT2otT3=UP1T1otUP2T2otUP3T3}
  U_{\mathrm{P}_1,\mathrm{P}_2,\mathrm{P}_3}(\ket{\Theta_1}\otimes\ket{\Theta_2}\otimes\ket{\Theta_3})
  =(U_{\mathrm{P}_1}\ket{\Theta_1})\otimes(U_{\mathrm{P}_2}\ket{\Theta_2})\otimes(U_{\mathrm{P}_3}\ket{\Theta_3})
\end{equation}
for every states $\ket{\Theta_1}$, $\ket{\Theta_2}$, and $\ket{\Theta_3}$ such that
$\ket{\Theta_i}$ is a state of the composite system
consisting of the systems $\mathcal{Q}_{\mathrm{P}_i3}$ and  $\mathcal{Q}_{\mathrm{P}_i}$
and the apparatus $\mathcal{A}_{\mathrm{P}_i3}$ and $\mathcal{A}_{\mathrm{P}_i4}$
for
all
$i=1,2,3$.
We use $\Omega$ to denote the alphabet $$\{0,1\}^3\times\{+1,-1\}^3.$$ Then,
for each
$\ket{\Psi_1}\in\mathcal{H}_{\mathrm{P}_13}\otimes\mathcal{H}_{\mathrm{P}_1}$,
$\ket{\Psi_2}\in\mathcal{H}_{\mathrm{P}_23}\otimes\mathcal{H}_{\mathrm{P}_2}$,
and $\ket{\Psi_3}\in\mathcal{H}_{\mathrm{P}_33}\otimes\mathcal{H}_{\mathrm{P}_3}$,
using \eqref{unitary-Pi-GHZ} and \eqref{eq:UP1P2P3T1otT2otT3=UP1T1otUP2T2otUP3T3} we see that
\begin{align*}
&U_{\mathrm{P}_1,\mathrm{P}_2,\mathrm{P}_3}\left(\ket{\Psi_1}\otimes\ket{\Psi_2}\otimes\ket{\Psi_3}\otimes\ket{\Phi^{\mathrm{init}}}\right) \\
&=\left(U_{\mathrm{P}_1}\left(\ket{\Psi_1}\otimes
\ket{\Phi_{\mathrm{P}_13}^{\mathrm{init}}}\otimes\ket{\Phi_{\mathrm{P}_14}^{\mathrm{init}}}\right)\right)
\otimes
\left(U_{\mathrm{P}_2}\left(\ket{\Psi_2}\otimes
\ket{\Phi_{\mathrm{P}_23}^{\mathrm{init}}}\otimes\ket{\Phi_{\mathrm{P}_24}^{\mathrm{init}}}\right)\right)
\otimes
\left(U_{\mathrm{P}_3}\left(\ket{\Psi_3}\otimes
\ket{\Phi_{\mathrm{P}_33}^{\mathrm{init}}}\otimes\ket{\Phi_{\mathrm{P}_34}^{\mathrm{init}}}\right)\right) \\
&=
\sum_{(c_1,c_2,c_3,m_1,m_2,m_3)\in\Omega}
\left(\left(E_{c_1}\otimes E_{c_2}\otimes E_{c_3}\otimes E^1_{c_1,m_1}\otimes E^2_{c_2,m_2}\otimes E^3_{c_3,m_3}\right)
\left(\ket{\Psi_1}\otimes\ket{\Psi_2}\otimes\ket{\Psi_3}\right)\right) \\
&\hspace{40mm}\otimes\ket{\Phi[c_1,c_2,c_3,m_1,m_2,m_3]},
\end{align*}
where
$\ket{\Phi^{\mathrm{init}}}$ denotes
$$\ket{\Phi_{\mathrm{P}_13}^{\mathrm{init}}}\otimes\ket{\Phi_{\mathrm{P}_23}^{\mathrm{init}}}\otimes\ket{\Phi_{\mathrm{P}_33}^{\mathrm{init}}}
\otimes
\ket{\Phi_{\mathrm{P}_14}^{\mathrm{init}}}\otimes\ket{\Phi_{\mathrm{P}_24}^{\mathrm{init}}}\otimes\ket{\Phi_{\mathrm{P}_34}^{\mathrm{init}}},$$
and
$\ket{\Phi[c_1,c_2,c_3,m_1,m_2,m_3]}$ denotes
$$\ket{\Phi_{\mathrm{P}_13}[c_1]}\otimes\ket{\Phi_{\mathrm{P}_23}[c_2]}\otimes\ket{\Phi_{\mathrm{P}_33}[c_3]}
\otimes
\ket{\Phi_{\mathrm{P}_14}[m_1]}\otimes\ket{\Phi_{\mathrm{P}_24}[m_2]}\otimes\ket{\Phi_{\mathrm{P}_34}[m_3]}$$
for each $(c_1,c_2,c_3,m_1,m_2,m_3)\in\Omega$.
It follows from the linearity of $U_{\mathrm{P}_1,\mathrm{P}_2,\mathrm{P}_3}$ that
\begin{align*}
  &U_{\mathrm{P}_1,\mathrm{P}_2,\mathrm{P}_3}\left(\ket{\Psi}\otimes\ket{\Phi^{\mathrm{init}}}\right)
  =\sum_{(c_1,c_2,c_3,m_1,m_2,m_3)\in\Omega}
  \left(\left(E_{c_1}\otimes E_{c_2}\otimes E_{c_3}\otimes E^1_{c_1,m_1}\otimes E^2_{c_2,m_2}\otimes E^3_{c_3,m_3}\right)\ket{\Psi}\right) \\
  &\hspace{80mm}\otimes\ket{\Phi[c_1,c_2,c_3,m_1,m_2,m_3]}
\end{align*}
for every
$\ket{\Psi}\in
\mathcal{H}_{\mathrm{P}_13}\otimes\mathcal{H}_{\mathrm{P}_23}\otimes\mathcal{H}_{\mathrm{P}_33}\otimes
\mathcal{H}_{\mathrm{P}_1}\otimes\mathcal{H}_{\mathrm{P}_2}\otimes\mathcal{H}_{\mathrm{P}_3}$.
This $U_{\mathrm{P}_1,\mathrm{P}_2,\mathrm{P}_3}$ describes the unitary time-evolution of
the parties~$\mathrm{P}_1,\mathrm{P}_2,\mathrm{P}_3$
in the repeated once of the infinite repetition of
that procedure in Protocol~\ref{GHZ}
which consists of
the six steps: Steps~$\mathrm{P}_i$3 and $\mathrm{P}_i$4
on the party $\mathrm{P}_i$'s side with all $i=1,2,3$.
Totally, prior to the application of $U_{\mathrm{P}_1,\mathrm{P}_2,\mathrm{P}_3}$,
the total system consisting of
$\mathcal{Q}_{\mathrm{P}_13}$, $\mathcal{Q}_{\mathrm{P}_23}$, $\mathcal{Q}_{\mathrm{P}_33}$,
$\mathcal{Q}_{\mathrm{P}_1}$, $\mathcal{Q}_{\mathrm{P}_2}$, and $\mathcal{Q}_{\mathrm{P}_3}$
is prepared in the state
\begin{equation}\label{GHZ-fps-calculation:eq04}
  \ket{\Psi^{\mathrm{init}}}:=\ket{+}\otimes\ket{+}\otimes\ket{+}\otimes\ket{\text{GHZ}}.
\end{equation}

\subsection{Application of the principle of typicality}

The operator $U_{\mathrm{P}_1,\mathrm{P}_2,\mathrm{P}_3}$ applying to the initial state $\ket{\Psi^{\mathrm{init}}}$ describes
the \emph{repeated once} of the infinite repetition of the measurements
in Protocol~\ref{GHZ},
where the execution of
Steps~$\mathrm{P}_i$3 and $\mathrm{P}_i$4
with all $i=1,2,3$
is infinitely repeated.
Actually, we can
check that a collection
\begin{equation}\label{GHZmos}
  \left\{E_{c_1}\otimes E_{c_2}\otimes E_{c_3}\otimes E^1_{c_1,m_1}\otimes E^2_{c_2,m_2}\otimes E^3_{c_3,m_3}\right\}_{(c_1,c_2,c_3,m_1,m_2,m_3)\in\Omega}
\end{equation}
forms \emph{measurement operators},
as is confirmed by the following calculation:
\begin{align*}
  &\sum_{(c_1,c_2,c_3,m_1,m_2,m_3)\in\Omega}
  \left(E_{c_1}\otimes E_{c_2}\otimes E_{c_3}\otimes E^1_{c_1,m_1}\otimes E^2_{c_2,m_2}\otimes E^3_{c_3,m_3}\right)^\dag \\
  &\hspace{40mm}\left(E_{c_1}\otimes E_{c_2}\otimes E_{c_3}\otimes E^1_{c_1,m_1}\otimes E^2_{c_2,m_2}\otimes E^3_{c_3,m_3}\right) \\
  &=\sum_{(c_1,c_2,c_3,m_1,m_2,m_3)\in\Omega}
  E_{c_1}\otimes E_{c_2}\otimes E_{c_3}\otimes E^1_{c_1,m_1}\otimes E^2_{c_2,m_2}\otimes E^3_{c_3,m_3} \\
  &=\sum_{c_1,c_2,c_3\in\{0,1\}}  E_{c_1}\otimes E_{c_2}\otimes E_{c_3}\otimes
  \left(\sum_{m_1=\pm 1}E^1_{c_1,m_1}\right)\otimes
  \left(\sum_{m_2=\pm 1}E^2_{c_2,m_2}\right)\otimes
  \left(\sum_{m_3=\pm 1}E^3_{c_3,m_3}\right) \\
  &=\left(\sum_{c_1=0,1}E_{c_1}\right)\otimes \left(\sum_{c_2=0,1}E_{c_2}\right) \otimes \left(\sum_{c_3=0,1}E_{c_3}\right)
  \otimes I_2\otimes I_2\otimes I_2 \\
  &=I_2\otimes I_2\otimes I_2\otimes I_2\otimes I_2\otimes I_2 \\
  &=I_{64},
\end{align*}
where the third equality follows from \eqref{GHZ-fps-calculation:eq03},
and the forth equality follows from \eqref{eq:Ec=kcbc}.
Thus,
the total application $U_{\mathrm{P}_1,\mathrm{P}_2,\mathrm{P}_3}$ of
$U_{\mathrm{P}_13}$ and $U_{\mathrm{P}_14}$,
$U_{\mathrm{P}_23}$ and $U_{\mathrm{P}_24}$,
and $U_{\mathrm{P}_33}$ and $U_{\mathrm{P}_34}$
can be regarded as
a
\emph{single measurement} which is described by
the measurement operators \eqref{GHZmos}
and whose all possible outcomes form the set $\Omega$.

Hence, we can apply Definition~\ref{pmrpwst}
to this scenario
of
the setting of measurements.
Therefore, according to Definition~\ref{pmrpwst},
we can see that a \emph{world} is an infinite sequence over
$\Omega$
and the probability measure induced by
the
\emph{probability measure representation for the prefixes of
worlds}
is
a Bernoulli measure $\lambda_P$ on
$\Omega^\infty$,
where $P$ is a finite probability space on
$\Omega$
such that
$P(c_1,c_2,c_3,m_1,m_2,m_3)$ is the square of the norm of
the
vector
\begin{equation*}
  \left(\left(E_{c_1}\otimes E_{c_2}\otimes E_{c_3}\otimes E^1_{c_1,m_1}\otimes E^2_{c_2,m_2}\otimes E^3_{c_3,m_3}\right)
  \ket{\Psi^{\mathrm{init}}}\right)\otimes\ket{\Phi[c_1,c_2,c_3,m_1,m_2,m_3]}
\end{equation*}
for every $(c_1,c_2,c_3,m_1,m_2,m_3)\in\Omega$.
Here $\Omega$ equals the set of all possible records of the apparatus
in the \emph{repeated once} of the measurements.

Let us calculate the explicit form of $P(c_1,c_2,c_3,m_1,m_2,m_3)$.
First, using \eqref{GHZ-fps-calculation:eq04}, \eqref{eq:ket-plus}, and \eqref{eq:Ec=kcbc},
we have that
\begin{equation}\label{GHZ-fps-calculation:eq05}
\begin{split}
  &P(c_1,c_2,c_3,m_1,m_2,m_3) \\
  &=\bra{+}E_{c_1}\ket{+}\bra{+}E_{c_2}\ket{+}\bra{+}E_{c_3}\ket{+}
  \bra{\text{GHZ}}\left(E^1_{c_1,m_1}\otimes E^2_{c_2,m_2}\otimes E^3_{c_3,m_3}\right)\ket{\text{GHZ}} \\
  &=\frac{1}{8}
  \bra{\text{GHZ}}\left(E^1_{c_1,m_1}\otimes E^2_{c_2,m_2}\otimes E^3_{c_3,m_3}\right)\ket{\text{GHZ}}
\end{split}
\end{equation}
for each $(c_1,c_2,c_3,m_1,m_2,m_3)\in\Omega$.
Then the inner-product
$\bra{\text{GHZ}}(E^1_{c_1,m_1}\otimes E^2_{c_2,m_2}\otimes E^3_{c_3,m_3})\ket{\text{GHZ}}$
above is calculated as follows:
Using \eqref{GHZ-fps-calculation:eq02} and \eqref{GHZ-fps-calculation:eq03}
we have that
\begin{equation}\label{GHZ-fps-calculation:eq06}
  E_{0,m}^k=\frac{1}{2}(I_2+m A^k_0)\quad\text{ and }\quad
  E_{1,m}^k=\frac{1}{2}(I_2+m A^k_1)
\end{equation}
for every $k=1,2,3$ and $m\in\{+1,-1\}$.
On the other hand, recall that
\begin{equation}\label{GHZ-fps-calculation:eq07}
  \ket{\text{GHZ}}
  =\frac{\ket{0}\otimes\ket{0}\otimes\ket{0}-\ket{1}\otimes\ket{1}\otimes\ket{1}}{\sqrt{2}}.
\end{equation}
Therefore,
since
$\bra{1}I_2\ket{0}=\bra{0}I_2\ket{1}=0$ and
$\bra{0}X\ket{0}=\bra{1}X\ket{1}=\bra{0}Y\ket{0}=\bra{1}Y\ket{1}=0$,
we have that
\begin{equation}\label{GHZ-fps-calculation:eq08}
  \bra{\text{GHZ}}(A^1_{c_1}\otimes I_2\otimes I_2)\ket{\text{GHZ}}
  =\bra{\text{GHZ}}(I_2\otimes A^2_{c_2}\otimes I_2)\ket{\text{GHZ}}
  =\bra{\text{GHZ}}(I_2\otimes I_2\otimes A^3_{c_3})\ket{\text{GHZ}}=0
\end{equation}
and
\begin{equation}\label{GHZ-fps-calculation:eq09}
  \bra{\text{GHZ}}(A^1_{c_1}\otimes A^2_{c_2}\otimes I_2)\ket{\text{GHZ}}
  =\bra{\text{GHZ}}(I_2\otimes A^2_{c_2}\otimes A^3_{c_3})\ket{\text{GHZ}}
  =\bra{\text{GHZ}}(A^1_{c_1}\otimes I_2\otimes A^3_{c_3})\ket{\text{GHZ}}=0
\end{equation}
for every $c_1,c_2,c_3\in\{0,1\}$.
Then,
for each $k=1,2,3$ and $c\in\{0,1\}$,
we see from \eqref{Bell-fps-calculation:eq02}
that $\bra{1}A^k_c\ket{0}=i^c$
and therefore $\bra{0}A^k_c\ket{1}=\overline{\bra{1}A^k_c\ket{0}}=(-i)^c$.
Thus,
since $\bra{0}X\ket{0}=\bra{1}X\ket{1}=\bra{0}Y\ket{0}=\bra{1}Y\ket{1}=0$,
using \eqref{GHZ-fps-calculation:eq07}
we have that
\begin{equation}\label{GHZ-fps-calculation:eq10}
\begin{split}
  &\bra{\text{GHZ}}(A^1_{c_1}\otimes A^2_{c_2}\otimes A^3_{c_3})\ket{\text{GHZ}} \\
  &=-\frac{1}{2}\left[\bra{1}A^1_{c_1}\ket{0}\bra{1}A^2_{c_2}\ket{0}\bra{1}A^3_{c_3}\ket{0}
  +\bra{0}A^1_{c_1}\ket{1}\bra{0}A^2_{c_2}\ket{1}\bra{0}A^3_{c_3}\ket{1}\right] \\
  &=-\frac{1}{2}\left[i^{c_1}i^{c_2}i^{c_3}+(-i)^{c_1}(-i)^{c_2}(-i)^{c_3}\right]
  =-\frac{1}{2}\left[e^{i\frac{\pi}{2}(c_1+c_2+c_3)}+e^{-i\frac{\pi}{2}(c_1+c_2+c_3)}\right] \\
  &=-\cos\left(\frac{\pi}{2}(c_1+c_2+c_3)\right)
\end{split}
\end{equation}
for each $c_1,c_2,c_3\in\{0,1\}$.
Thus, it follows from \eqref{GHZ-fps-calculation:eq06},
\eqref{GHZ-fps-calculation:eq08}, \eqref{GHZ-fps-calculation:eq09},
and \eqref{GHZ-fps-calculation:eq10} that
\[
  \bra{\text{GHZ}}\left(E^1_{c_1,m_1}\otimes E^2_{c_2,m_2}\otimes E^3_{c_3,m_3}\right)\ket{\text{GHZ}}
  =\frac{1}{8}\left[1-m_1 m_2 m_3\times\cos\left(\frac{\pi}{2}(c_1+c_2+c_3)\right)\right]
\]
for every $(c_1,c_2,c_3,m_1,m_2,m_3)\in\Omega$.
Hence, using \eqref{GHZ-fps-calculation:eq05},
the explicit form of $P(c_1,c_2,c_3,m_1,m_2,m_3)$ is obtained as follows:
\begin{equation}\label{Pcccmmm=1over64cm-GHZ}
  P(c_1,c_2,c_3,m_1,m_2,m_3)
  =\frac{1}{64}\left[1-m_1 m_2 m_3\times\cos\left(\frac{\pi}{2}(c_1+c_2+c_3)\right)\right]
\end{equation}
for every $(c_1,c_2,c_3,m_1,m_2,m_3)\in\Omega$.

Now, let us apply
Postulate~\ref{POT}, the \emph{principle of typicality},
to the setting of measurements
developed above.
Let $\alpha$ be \emph{our world} in the infinite repetition of the measurements in the above setting.
This $\alpha$
is an infinite sequence over $\Omega$
consisting
of records in the apparatuses
which is being generated by the infinite repetition of the measurements
described by the measurement operators~\eqref{GHZmos}
in the above setting.
Since the Bernoulli measure $\lambda_P$ on $\Omega^\infty$ is
the probability measure induced by the
probability measure representation
for the prefixes of
worlds
in the above setting,
it follows from
Postulate~\ref{POT} that
\emph{$\alpha$ is Martin-L\"of $P$-random}.

\subsection{\boldmath Refined derivation of the perfect correlations over
the three parties~$\mathrm{P}_1, \mathrm{P}_2, \mathrm{P}_3$}

For each $c_1,c_2,c_3\in\{0,1\}$, we use $H(c_1,c_2,c_3)$ to denote the set
$$\left\{(c_1,c_2,c_3,m_1,m_2,m_3)\,\middle\vert\, m_1,m_2,m_3\in\{+1,-1\}\right\}.$$

Now,
let $c_1,c_2,c_3\in\{0,1\}$.
The set $H(c_1,c_2,c_3)$ consists of all records of the apparatuses
$\mathcal{A}_{\mathrm{P}_13}$, $\mathcal{A}_{\mathrm{P}_23}$,
$\mathcal{A}_{\mathrm{P}_33}$,
$\mathcal{A}_{\mathrm{P}_14}$, $\mathcal{A}_{\mathrm{P}_24}$,
and $\mathcal{A}_{\mathrm{P}_34}$,
in a repeated once of the procedure in Protocol~\ref{GHZ}, where
\emph{the party~$\mathrm{P}_i$ gets the outcome $c_i$ in Step~$\mathrm{P}_i$3 for every $i=1,2,3$}.
It follows from
\eqref{Pcccmmm=1over64cm-GHZ}
that
$$P(H(c_1,c_2,c_3))=\frac{1}{8},$$
as expected from the point of view of the conventional quantum mechanics,
and
moreover
\begin{equation}\label{PHc3c3m3-GHZ}
\begin{split}
  P_{H(c_1,c_2,c_3)}(c_1,c_2,c_3,m_1,m_2,m_3)
  &=\frac{P(c_1,c_2,c_3,m_1,m_2,m_3)}{P(H(c_1,c_2,c_3))} \\
  &=\frac{1}{8}\left[1-m_1 m_2 m_3\times\cos\left(\frac{\pi}{2}(c_1+c_2+c_3)\right)\right]
\end{split}
\end{equation}
for every $m_1,m_2,m_3\in\{+1,-1\}$.
Here we use the notation presented in Section~\ref{subsec:Conditional probability}.
Let $$\alpha_{c_1,c_2,c_3}:=\cond{H(c_1,c_2,c_3)}{\alpha}.$$
Then, since $\alpha$ is Martin-L\"of $P$-random,
using Theorem~\ref{conditional_probability} we have that
$\alpha_{c_1,c_2,c_3}$ is Martin-L\"of $P_{H(c_1,c_2,c_3)}$-random for the finite probability space $P_{H(c_1,c_2,c_3)}$ on $H(c_1,c_2,c_3)$.
Recall that $\cond{H(c_1,c_2,c_3)}{\alpha}$ is defined as an infinite sequence
over the alphabet $H(c_1,c_2,c_3)$
obtained from
$\alpha$ by eliminating all elements of $\Omega\setminus H(c_1,c_2,c_3)$ occurring in $\alpha$.
In other words, $\alpha_{c_1,c_2,c_3}$ is the subsequence of $\alpha$
consisting
only
of results that
the party~$\mathrm{P}_i$ gets the outcome $c_i$ in Step~$\mathrm{P}_i$3 for every $i=1,2,3$.
For each $k\in\N^+$, we denote the $k$th element $\alpha_{c_1,c_2,c_3}(k)$ of
the subsequence $\alpha_{c_1,c_2,c_3}$ as
\begin{equation}\label{eq:alphc3=c3m3}
  \alpha_{c_1,c_2,c_3}(k)
  =(c_1,c_2,c_3,m_1^{c_1,c_2,c_3}(k),m_2^{c_1,c_2,c_3}(k),m_3^{c_1,c_2,c_3}(k)).
\end{equation}
In this manner, we introduce infinite sequences $m_1^{c_1,c_2,c_3}$, $m_2^{c_1,c_2,c_3}$, and $m_3^{c_1,c_2,c_3}$ over $\{+1,-1\}$.
For each $i=1,2,3$, the sequence $m_i^{c_1,c_2,c_3}$
is
\emph{the infinite sequence of outcomes of the
measurements
performed
by the party~$\mathrm{P}_i$  in Step~$\mathrm{P}_i$4
over the infinite repetition of the procedure in Protocol~\ref{GHZ} in our world,
conditioned that
the party~$\mathrm{P}_j$ gets the outcome $c_j$ in Step~$\mathrm{P}_j$3 for every $j=1,2,3$}.

We can then show the following theorem
which reveals perfect correlations of measurement results over
the three parties~$\mathrm{P}_1, \mathrm{P}_2, \mathrm{P}_3$:

\begin{theorem}[Perfect correlations over the three parties]\label{thm:GHZ}\
\begin{enumerate}
\item If the string $c_1c_2c_3$ is equal to either $011$, $101$, or $110$, then
\[
m_1^{c_1,c_2,c_3}(k)m_2^{c_1,c_2,c_3}(k)m_3^{c_1,c_2,c_3}(k)=+1
\]
for every $k\in\N^+$.
\item If the string $c_1c_2c_3$ is equal to $000$, then
\[
m_1^{c_1,c_2,c_3}(k)m_2^{c_1,c_2,c_3}(k)m_3^{c_1,c_2,c_3}(k)=-1
\]
for every $k\in\N^+$.
\end{enumerate}
\end{theorem}

\begin{proof}
Let $c_1,c_2,c_3\in\{0,1\}$.
On the one hand, since $\alpha_{c_1,c_2,c_3}$ is Martin-L\"of $P_{H(c_1,c_2,c_3)}$-random,
it follows from Corollary~\ref{cor:always-positive-probability} that
\begin{equation}\label{eq:PHc3alphc3>0-GHZ}
  P_{H(c_1,c_2,c_3)}(\alpha_{c_1,c_2,c_3}(k))>0
\end{equation}
for all $k\in\N^+$.
On the other hand, using \eqref{PHc3c3m3-GHZ} and \eqref{eq:alphc3=c3m3} we have that
\begin{equation}\label{eq:PHc3alphc3=f181-cosc3m3-GHZ}
  P_{H(c_1,c_2,c_3)}(\alpha_{c_1,c_2,c_3}(k))=
  \frac{1}{8}\left[1-m_1^{c_1,c_2,c_3}(k)m_2^{c_1,c_2,c_3}(k)m_3^{c_1,c_2,c_3}(k)\times\cos\left(\frac{\pi}{2}(c_1+c_2+c_3)\right)\right]
\end{equation}
for every $k\in\N^+$.
Thus it follows from \eqref{eq:PHc3alphc3>0-GHZ} and \eqref{eq:PHc3alphc3=f181-cosc3m3-GHZ} that
for every $c_1,c_2,c_3\in\{0,1\}$ and $k\in\N^+$ it holds that
\begin{equation}\label{eq:cosc3m3<1-GHZ}
  m_1^{c_1,c_2,c_3}(k)m_2^{c_1,c_2,c_3}(k)m_3^{c_1,c_2,c_3}(k)\times\cos\left(\frac{\pi}{2}(c_1+c_2+c_3)\right)<1.
\end{equation}
Note here that
\begin{equation}\label{eq:m3=pm1-GHZ}
  m_1^{c_1,c_2,c_3}(k)m_2^{c_1,c_2,c_3}(k)m_3^{c_1,c_2,c_3}(k)=\pm 1
\end{equation}
for every $k\in\N^+$.

(i) Let $c_1,c_2,c_3\in\{0,1\}$ with $c_1+c_2+c_3=2$.
Then it follows from \eqref{eq:cosc3m3<1-GHZ} and \eqref{eq:m3=pm1-GHZ} that
\[
  m_1^{c_1,c_2,c_3}(k)m_2^{c_1,c_2,c_3}(k)m_3^{c_1,c_2,c_3}(k)=+1
\]
for every $k\in\N^+$.

(ii) Let $c_1,c_2,c_3\in\{0,1\}$ with $c_1+c_2+c_3=0$, i.e., $c_1=c_2=c_3=0$.
Then it follows from \eqref{eq:cosc3m3<1-GHZ} and \eqref{eq:m3=pm1-GHZ} that
\[
  m_1^{c_1,c_2,c_3}(k)m_2^{c_1,c_2,c_3}(k)m_3^{c_1,c_2,c_3}(k)=-1
\]
for every $k\in\N^+$.
\end{proof}

The result of Theorem~\ref{thm:GHZ} is certainly a \emph{refinement} of
the perfect correlations of measurement results over three parties
$\mathrm{P}_1$, $\mathrm{P}_2$, $\mathrm{P}_3$
which are expected from the aspect of the conventional quantum mechanics
as its prediction.

\section{Refined demonstration of underivability of the perfect correlations from local realism}
\label{sec:GHZ-imperfect}

Mermin~\cite{Mer90} provides an analysis for Protocol~\ref{GHZ} which
demonstrates
that
the assumptions of local realism cannot recover
the prediction of quantum mechanics, i.e., the perfect correlations of measurement results over the three parties~$\mathrm{P}_1, \mathrm{P}_2, \mathrm{P}_3$.
This Mermin's work is again based on the
work of Greenberger, Horne, and Zeilinger~\cite{GHZ89}.
In this section,
we \emph{refine} and \emph{reformulate}
Mermin's
analysis and
thereby
the assumptions of local realism,
in the framework of
the \emph{operational characterization of the notion of probability
by algorithmic randomness}~\cite{T14,T15,T16arXiv} (see Section~\ref{subsec:OCNP}).

\subsection{Refinement of the assumptions of local realism}

In what follows, we refine and reformulate
``the common sense analysis'' for Protocol~\ref{GHZ}
based on the assumptions of local realism,
in the framework of the operational characterization of the notion of probability
reviewed in Section~\ref{subsec:OCNP}.
Basically, we follow the flow of the argument of Mermin~\cite{Mer90},
while refining it appropriately
in terms of our operational characterization of the notion of probability.
We also proceed according to Nielsen and Chuang~\cite[Section~2.6]{NC00}
regarding how to introduce the assumptions of local realism.

We first forget all the knowledge of quantum mechanics.
We analyze Protocol~\ref{GHZ}
based on  ``our common sense notions of how the world works,''
as in Section~\ref{sec-Bell_inequality}.
This leads to a contradiction with the prediction by quantum mechanics,
i.e., Theorem~\ref{thm:GHZ}.
Thus, we perform ``the common sense analysis'' for Protocol~\ref{GHZ}.
In doing so, we are implicitly assuming the following two assumptions:
\begin{description}
\item[The assumption of realism:]
The assumption that
the observables~$A^1_0$,  $A^1_1$, $A^2_0$, $A^2_1$, $A^3_0$, $A^3_1$ have
definite values~$m^1_0$, $m^1_1$, $m^2_0$, $m^2_1$, $m^3_0$, $m^3_1$, respectively,
which exist independent of observation.
\item[The assumption of locality:]
The assumption that
any
party performing its measurement does not influence
the results of measurements performed by the other
two
parties
among the three parties~$\mathrm{P}_1, \mathrm{P}_2, \mathrm{P}_3$.
\end{description}
These two assumptions together form the \emph{assumptions of local realism}.

Now,
according to Mermin~\cite{Mer90}
and Nielsen and Chuang~\cite[Section~2.6]{NC00},
let us make ``the common sense analysis'' for Protocol~\ref{GHZ},
based on the assumptions of local realism.
In Protocol~\ref{GHZ},
for each $i=1,2,3$, the party~$\mathrm{P}_i$
performs the measurement of either the observable $A^i_0$ or $A^i_1$
over the particle~$i$ in Step $\mathrm{P}_i$4
to obtain outcome $m_i\in\{+1,-1\}$, depending on $c_i=0$ or $1$.
Based on the assumptions of local realism,
we
assume
that each of the six observables
$A^1_0$,  $A^1_1$, $A^2_0$, $A^2_1$, $A^3_0$, and $A^3_1$ 
has a specific value before the measurement, which is merely revealed by the measurement.
In particular,
\emph{in the terminology of the conventional probability theory},
we assume that
$$p(m^1_0, m^1_1, m^2_0, m^2_1, m^3_0, m^3_1)$$
is the ``probability'' that,
before the measurements are performed,
the system is in a state where
$A^1_0=m^1_0$,  $A^1_1=m^1_1$, $A^2_0=m^2_0$, $A^2_1=m^2_1$, $A^3_0=m^3_0$, and $A^3_1=m^3_1$.
This ``probability''
may depend on how the party $\mathrm{P}_0$ prepares the three spin-$1/2$ particles,
and on experimental noise.
In the framework of the operational characterization of the notion of probability,
\emph{based on Thesis~\ref{thesis}},
the assumption above is refined and reformulated in the following form:

\begin{assumption}\label{ASLR1-GHZ}
Let $\omega$ be an infinite sequence of the values
$(m^1_0, m^1_1, m^2_0, m^2_1, m^3_0, m^3_1)$
of the observables $A^1_0$, $A^1_1$, $A^2_0$, $A^2_1$, $A^3_0$,
$A^3_1$ 
which is being generated by
the infinite repetition of the procedure
in Protocol~\ref{GHZ}.
Then there exists a finite probability space $P$ on $\{+1,-1\}^6$ such that
$\omega$ is a \emph{Martin-L\"of $P$-random} infinite sequence over $\{+1,-1\}^6$.
\qed
\end{assumption}

Assumption~\ref{ASLR1-GHZ} is an \emph{operational refinement} of one of
the consequences
of the assumptions of local realism.
In Protocol~\ref{GHZ}, for each $i=1,2,3$,
the party $\mathrm{P}_i$ tosses a fair coin $C_i$ to get outcome $c_i\in\{0,1\}$
in Step~$\mathrm{P}_i$3.
In the framework of the operational characterization of the notion of probability,
these probabilistic phenomena are refined and reformulated
in the following form:

\begin{assumption}\label{Impl-Thesis-GHZ}
Let $i\in\{1,2,3\}$,
and let $\gamma_i$ be an infinite binary sequence which is
being generated by infinitely repeated
tossing
of the fair coin $C_i$ by the party~$\mathrm{P}_i$
in Protocol~\ref{GHZ}.
Then the infinite sequence $\gamma_i$ is
a \emph{Martin-L\"of $U$-random} sequence over $\{0,1\}$,
where $U$ is a finite probability space on $\{0,1\}$ such that $U(0)=U(1)=1/2$.
Namely, $\gamma_i$ is \emph{Martin-L\"of random}.
\qed
\end{assumption}

Assumption~\ref{Impl-Thesis-GHZ} is just an implementation of Thesis~\ref{thesis}
in an infinite repetition of tossing of a fair coin.
In order to advance
``the common sense analysis'' for Protocol~\ref{GHZ}
further
in a rigorous manner,
however,
we need to make
an additional
assumption for the relation
among the infinite sequences $\omega$, $\gamma_1$, $\gamma_2$, and $\gamma_3$.
Namely,
the infinite sequences $\omega$, $\gamma_1$, $\gamma_2$, and $\gamma_3$
need to be \emph{independent}.
Thus, based on the notion of independence given in Definition~\ref{independency-of-ensembles},
we assume the following:

\begin{assumption}\label{AL2-GHZ}
The infinite sequences $\omega$, $\gamma_1$, $\gamma_2$, and $\gamma_3$
are \emph{independent}.
\qed
\end{assumption}

Assumption~\ref{AL2-GHZ} is an \emph{operational refinement} of one of
the consequences
of the assumption of locality.

Now, according to Definition~\ref{independency-of-ensembles},
Assumption~\ref{ASLR1-GHZ}, Assumption~\ref{Impl-Thesis-GHZ}, and Assumption~\ref{AL2-GHZ}
together
imply
the following \emph{single} assumption:

\begin{assumption}\label{ALRsingle-GHZ}
Let $\omega$ be an infinite sequence of the values
$(m^1_0, m^1_1, m^2_0, m^2_1, m^3_0, m^3_1)$
of the observables $A^1_0$, $A^1_1$, $A^2_0$, $A^2_1$, $A^3_0$,
$A^3_1$ 
which is being generated by
the infinite repetition of the procedure
in Protocol~\ref{GHZ}.
For each $i\in\{1,2,3\}$,
let $\gamma_i$ be an infinite binary sequence which is
being generated by infinitely repeated
tossing
of the fair coin $C_i$ by the party~$\mathrm{P}_i$
in Protocol~\ref{GHZ}.
Then there exists a finite probability space $P$ on $\{+1,-1\}^6$ such that
the infinite sequence $\omega\times\gamma_1\times\gamma_2\times\gamma_3$ over
$\{+1,-1\}^6\times\{0,1\}\times\{0,1\}\times\{0,1\}$
is \emph{Martin-L\"of $P\times U\times U\times U$-random},
where $U$ is a finite probability space on $\{0,1\}$ such that $U(0)=U(1)=1/2$.
\qed
\end{assumption}

However, it follows from Theorem~\ref{independency-imlplies-each-randomness} that
if the infinite sequence $\omega\times\gamma_1\times\gamma_2\times\gamma_3$ over
$\{+1,-1\}^6\times\{0,1\}\times\{0,1\}\times\{0,1\}$
is Martin-L\"of $P\times U\times U\times U$-random then
the infinite sequence $\omega$ over $\{+1,-1\}^6$ is Martin-L\"of $P$-random and
the infinite binary sequences $\gamma_1$, $\gamma_2$, and $\gamma_3$ are Martin-L\"of $U$-random.
Thus, in fact, Assumption~\ref{ALRsingle-GHZ} is \emph{equivalent} to
the conjunction of
Assumption~\ref{ASLR1-GHZ}, Assumption~\ref{Impl-Thesis-GHZ}, and Assumption~\ref{AL2-GHZ}.
Hence, Assumption~\ref{ALRsingle-GHZ}
\emph{alone} serves
as an \emph{operational refinement} of the whole of the assumptions of local realism in the context of Protocol~\ref{GHZ}.

\begin{remark}
As in Remark~\ref{rem:RC-Bell} on Assumption~\ref{ALRsingle}
for ``the common sense analysis'' of Protocol~\ref{Bell},
based on Theorem~42 of Tadaki~\cite{T16arXiv},
in Assumption~\ref{ALRsingle-GHZ}
the statement
``the infinite sequence $\omega\times\gamma_1\times\gamma_2\times\gamma_3$ over
$\{+1,-1\}^6\times\{0,1\}\times\{0,1\}\times\{0,1\}$
is Martin-L\"of $P\times U\times U\times U$-random''
can be \emph{equivalently} rephrased
as the statement that
\begin{enumerate}
  \item the infinite sequence $\omega$ is Martin-L\"of $P$-random,
  \item the infinite sequence $\gamma_1$ is Martin-L\"of random
    \emph{relative to $\omega$},
  \item the infinite sequence $\gamma_2$ is Martin-L\"of random
    \emph{relative to $\omega$ and $\gamma_1$}, and
  \item the infinite sequence $\gamma_3$ is Martin-L\"of random
    \emph{relative to $\omega$, $\gamma_1$, and $\gamma_2$}.
\end{enumerate}
Note that in establishing this equivalence
we do not
have
to impose any computability restrictions on $P$ at all:
In Assumption~\ref{ALRsingle-GHZ} (or in Assumption~\ref{ASLR1-GHZ})
the finite probability space $P$ on $\{+1,-1\}^6$ can be chosen completely arbitrarily
without failing this equivalence (see Theorem~42 of Tadaki~\cite{T16arXiv}).
\qed
\end{remark}

\subsection{Refined demonstration of underivability of the perfect correlations}

Based on Assumption~\ref{ALRsingle-GHZ}
(or equivalently, the conjunction of Assumptions~\ref{ASLR1-GHZ}, \ref{Impl-Thesis-GHZ}, and \ref{AL2-GHZ}),
let us make a contradiction with Theorem~\ref{thm:GHZ}
in the framework of the operational characterization of the notion of probability.
Let $$\alpha:=\omega\times\gamma_1\times\gamma_2\times\gamma_3.$$
Then, $\alpha$ is Martin-L\"of $P\times U\times U\times U$-random by Assumption~\ref{ALRsingle-GHZ}.
On the other hand,
note that
\begin{equation}\label{PtUtUm6=f18m6-Local-Realism-GHZ}
  (P\times U\times U\times U)((m^1_0, m^1_1, m^2_0, m^2_1, m^3_0, m^3_1),c_1,c_2,c_3)
  =\frac{1}{8}P(m^1_0, m^1_1, m^2_0, m^2_1, m^3_0, m^3_1)
\end{equation}
for every $m^1_0, m^1_1, m^2_0, m^2_1, m^3_0, m^3_1\in\{+1,-1\}$ and
every $c_1,c_2,c_3\in\{0,1\}$.
Here we use the notation presented in Section~\ref{sec-IMLP}.
For each $c_1,c_2,c_3\in\{0,1\}$, we use $B(c_1,c_2,c_3)$ to denote the set
$$\left\{(x,c_1,c_2,c_3)\,\middle\vert\, x\in\{+1,-1\}^6\right\}.$$

Now,
let $c_1,c_2,c_3\in\{0,1\}$.
The set $B(c_1,c_2,c_3)$ consists of all possible
results
in a repeated once of the procedure in Protocol~\ref{GHZ}, where
\emph{the party $\mathrm{P}_i$ gets the outcome
$c_i$
in Step~$\mathrm{P}_i$3 for every $i=1,2,3$}.
\emph{In the terminology of the conventional probability theory},
$(P\times U\times U\times U)(B(c_1,c_2,c_3))$ is
the ``probability'' that
\emph{the party $\mathrm{P}_i$ gets the outcome
$c_i$
in Step~$\mathrm{P}_i$3 for every $i=1,2,3$}.
Actually, it follows from \eqref{PtUtUm6=f18m6-Local-Realism-GHZ} that
$$(P\times U\times U\times U)(B(c_1,c_2,c_3))=\frac{1}{8},$$
as expected from the conventional probability theory.
Thus, it
follows from
\eqref{PtUtUm6=f18m6-Local-Realism-GHZ}
again that
\begin{equation}\label{PU3Bc3xc3=px-Local-Realism-GHZ}
  (P\times U\times U\times U)_{B(c_1,c_2,c_3)}(x,c_1,c_2,c_3)
  =\frac{(P\times U\times U\times U)(x,c_1,c_2,c_3)}{(P\times U\times U\times U)(B(c_1,c_2,c_3))}
  =P(x)
\end{equation}
for every $x\in\{+1,-1\}^6$.
Here we use the notation presented in Section~\ref{subsec:Conditional probability}.
Let $$\alpha_{c_1,c_2,c_3}:=\cond{B(c_1,c_2,c_3)}{\alpha}.$$
Then, since $\alpha$ is Martin-L\"of $P\times U\times U\times U$-random,
using Theorem~\ref{conditional_probability} we have that
$\alpha_{c_1,c_2,c_3}$ is
Martin-L\"of $(P\times U\times U\times U)_{B(c_1,c_2,c_3)}$-random
for the finite probability space $(P\times U\times U\times U)_{B(c_1,c_2,c_3)}$
on $B(c_1,c_2,c_3)$.
Recall that $\cond{B(c_1,c_2,c_3)}{\alpha}$ is defined as an infinite sequence
over the alphabet $B(c_1,c_2,c_3)$
obtained from
$\alpha$ by eliminating all elements of
$(\{+1,-1\}^6\times\{0,1\}\times\{0,1\}\times\{0,1\})\setminus B(c_1,c_2,c_3)$
occurring in $\alpha$.
In other words, $\alpha_{c_1,c_2,c_3}$ is the subsequence of $\alpha$
consisting of results that
the party $\mathrm{P}_i$ gets the outcome
$c_i$
in Step~$\mathrm{P}_i$3 for every $i=1,2,3$.
For each $k\in\N^+$, we denote the $k$th element
$\alpha_{c_1,c_2,c_3}(k)$
of the
subsequence
$\alpha_{c_1,c_2,c_3}$
as
\[
  \alpha_{c_1,c_2,c_3}(k)=
  ((m_{1,0}^{c_1,c_2,c_3}(k),m_{1,1}^{c_1,c_2,c_3}(k),
    m_{2,0}^{c_1,c_2,c_3}(k),m_{2,1}^{c_1,c_2,c_3}(k),
    m_{3,0}^{c_1,c_2,c_3}(k),m_{3,1}^{c_1,c_2,c_3}(k)),
    c_1,c_2,c_3).
\]
In this manner, we introduce six infinite sequences
\[
  m_{1,0}^{c_1,c_2,c_3}, m_{1,1}^{c_1,c_2,c_3},
  m_{2,0}^{c_1,c_2,c_3}, m_{2,1}^{c_1,c_2,c_3},
  m_{3,0}^{c_1,c_2,c_3}, m_{3,1}^{c_1,c_2,c_3}
\]
over $\{+1,-1\}$.
For each $i=1,2,3$ and $c=0,1$, the sequence $m_{i,c}^{c_1,c_2,c_3}$
is
\emph{the infinite sequence of
values of the observable $A^i_{c}$
over the infinite repetition of the procedure in Protocol~\ref{GHZ},
conditioned that
the party $\mathrm{P}_j$ gets the outcome $c_j$ in Step~$\mathrm{P}_j$3 for every $j=1,2,3$}.
Since the subsequence $\alpha_{c_1,c_2,c_3}$ is
Martin-L\"of $(P\times U\times U\times U)_{B(c_1,c_2,c_3)}$-random,
using Theorem~\ref{FI} and \eqref{PU3Bc3xc3=px-Local-Realism-GHZ}
we have that for every $x\in\{+1,-1\}^6$ it holds that
\begin{equation}\label{eq:LLN-GHZ-RL}
  \lim_{L\to\infty}\frac{N_{(x,c_1,c_2,c_3)}(\rest{\alpha_{c_1,c_2,c_3}}{L})}{L}=P(x).
\end{equation}
Recall here that $N_{(x,c_1,c_2,c_3)}(\rest{\alpha_{c_1,c_2,c_3}}{L})$ denotes the number of the occurrences of $(x,c_1,c_2,c_3)$
in the prefix of $\alpha_{c_1,c_2,c_3}$ of length $L$.

In what follows, in a refined manner in terms of 
our framework of the operational characterization of the notion of probability,
we
will
show that
``the common sense analysis'' developed
so far for Protocol~\ref{GHZ}
cannot recover
the \emph{perfect correlations}
revealed by Theorem~\ref{thm:GHZ},
which is
proven
by the refined argument of quantum mechanics
based on the principle of typicality.
Now,
assume contrarily that
``the common sense analysis''
developed so far
\emph{can} recover the perfect correlations
revealed by Theorem~\ref{thm:GHZ}.
That is, assume that the following (i) and (ii) simultaneously hold:
\begin{enumerate}
\item If the string $c_1c_2c_3$ is equal to either $011$, $101$, or $110$, then
\begin{equation}\label{eq:ms-GHZ-QM-plus}
  m_{1,c_1}^{c_1,c_2,c_3}(k)m_{2,c_2}^{c_1,c_2,c_3}(k)m_{3,c_3}^{c_1,c_2,c_3}(k)=+1
\end{equation}
for every $k\in\N^+$.
\item If the string $c_1c_2c_3$ is equal to $000$, then
\begin{equation}\label{eq:ms-GHZ-QM-minus}
  m_{1,c_1}^{c_1,c_2,c_3}(k)m_{2,c_2}^{c_1,c_2,c_3}(k)m_{3,c_3}^{c_1,c_2,c_3}(k)=-1
\end{equation}
for every $k\in\N^+$.
\end{enumerate}
Notice here that,
for every $c_1,c_2,c_3\in\{0,1\}$, the infinite sequences
\[
m_1^{c_1,c_2,c_3},
m_2^{c_1,c_2,c_3},
m_3^{c_1,c_2,c_3}
\]
over $\{+1,-1\}$ in Theorem~\ref{thm:GHZ} correspond to the infinite sequences
\[
m_{1,c_1}^{c_1,c_2,c_3},
m_{2,c_2}^{c_1,c_2,c_3},
m_{3,c_3}^{c_1,c_2,c_3}
\]
over $\{+1,-1\}$ in this section, respectively,
where the former is about quantum mechanics while the latter is about local realism.

Now,
let $c_1c_2c_3$ be equal to either $011$, $101$, or $110$.
For an arbitrary $x=(m^1_0, m^1_1, m^2_0, m^2_1, m^3_0, m^3_1)\in\{+1,-1\}^6$,
assume that $m^1_{c_1} m^2_{c_2} m^3_{c_3}=-1$.
Then, by the assumption~\eqref{eq:ms-GHZ-QM-plus}, we see that
$(x,c_1,c_2,c_3)\neq\alpha_{c_1,c_2,c_3}(k)$ for every $k\in\N^+$,
and therefore $N_{(x,c_1,c_2,c_3)}(\rest{\alpha_{c_1,c_2,c_3}}{L})=0$ for every $L\in\N^+$.
Thus using \eqref{eq:LLN-GHZ-RL} we have that $P(x)=0$.
In summary,
for every $(m^1_0, m^1_1, m^2_0, m^2_1, m^3_0, m^3_1)\in\{+1,-1\}^6$,
if $m^1_0 m^2_1 m^3_1=-1$, $m^1_1 m^2_0 m^3_1=-1$, or $m^1_1 m^2_1 m^3_0=-1$ then
\begin{equation}\label{eq:Pm6-GHZ-plus-zero1}
  P(m^1_0, m^1_1, m^2_0, m^2_1, m^3_0, m^3_1)=0.
\end{equation}
Note that for every $(m^1_0, m^1_1, m^2_0, m^2_1, m^3_0, m^3_1)\in\{+1,-1\}^6$
it holds that
$$m^1_0 m^2_0 m^3_0=(m^1_0 m^2_1 m^3_1)(m^1_1 m^2_0 m^3_1)(m^1_1 m^2_1 m^3_0).$$
Therefore, it follows from \eqref{eq:Pm6-GHZ-plus-zero1} that
for every $(m^1_0, m^1_1, m^2_0, m^2_1, m^3_0, m^3_1)\in\{+1,-1\}^6$
if $m^1_0 m^2_0 m^3_0=-1$ then
\begin{equation}\label{eq:Pm6-GHZ-plus-zero2}
  P(m^1_0, m^1_1, m^2_0, m^2_1, m^3_0, m^3_1)=0.
\end{equation}

On the other hand, let $c_1c_2c_3$ be equal to $000$.
For an arbitrary $x=(m^1_0, m^1_1, m^2_0, m^2_1, m^3_0, m^3_1)\in\{+1,-1\}^6$,
assume that
$m^1_{c_1} m^2_{c_2} m^3_{c_3}=+1$.
Then, by the assumption~\eqref{eq:ms-GHZ-QM-minus}, we see that
$(x,c_1,c_2,c_3)\neq\alpha_{c_1,c_2,c_3}(k)$ for every $k\in\N^+$,
and therefore $N_{(x,c_1,c_2,c_3)}(\rest{\alpha_{c_1,c_2,c_3}}{L})=0$ for every $L\in\N^+$.
Thus using \eqref{eq:LLN-GHZ-RL} we have that $P(x)=0$.
In summary,
for every $(m^1_0, m^1_1, m^2_0, m^2_1, m^3_0, m^3_1)\in\{+1,-1\}^6$,
if $m^1_0 m^2_0 m^3_0=+1$ then
\begin{equation}\label{eq:Pm6-GHZ-minus-zero2}
  P(m^1_0, m^1_1, m^2_0, m^2_1, m^3_0, m^3_1)=0.
\end{equation}

Finally, using \eqref{eq:Pm6-GHZ-plus-zero2} and \eqref{eq:Pm6-GHZ-minus-zero2}
we have that $P(x)=0$ for every $x\in\{+1,-1\}^6$.
However, this contradicts Assumption~\ref{ALRsingle-GHZ}
(or Assumption~\ref{ASLR1-GHZ}) which states that
$P$ is a finite probability space on $\{+1,-1\}^6$.
Thus, the two assumptions~\eqref{eq:ms-GHZ-QM-plus} and
\eqref{eq:ms-GHZ-QM-minus} do not hold simultaneously.
Therefore,
``the common sense analysis,'' i.e., local realism,
developed
so far
in this section
\emph{cannot} recover the \emph{perfect correlations}
revealed by Theorem~\ref{thm:GHZ}
as a prediction of quantum mechanics.

Hence, in a refined manner in terms of 
our framework of the operational characterization of the notion of probability,
we have
demonstrated
that
``the common sense analysis'' developed
for Protocol~\ref{GHZ}
cannot recover
the perfect correlations of measurement results over
the three parties~$\mathrm{P}_1, \mathrm{P}_2, \mathrm{P}_3$,
which is
rigorously revealed
in Section~\ref{sec:GHZ-QM} by the refined argument of quantum mechanics based on the principle of typicality.

\section{Conclusion}
\label{sec-Concluding_remarks}

In this paper, we have refined and reformulated
the argument of
local realism
versus quantum mechanics by algorithmic randomness.
On the one hand,
we have refined and reformulated local realism,
based on our operational characterization of the notion of probability by algorithmic randomness.
On the other hand,
we have refined and reformulated the corresponding argument of quantum mechanics
in our
rigorous
framework of quantum mechanics based on the principle of typicality.
In this paper, we have done
these
``two-sided refinements''
in the two contexts of
Bell's inequality and GHZ experiment.
Consequently,
in terms of
algorithmic randomness,
we have refined the derivation of the following fact:
Local realism
cannot recover the prediction of quantum mechanics.

The principle of typicality is a \emph{unified principle} which refines
Postulate~\ref{Born-rule}
and its related postulates about quantum measurements
\emph{in one lump} (see Tadaki~\cite{T18arXiv}).
In this paper, we have successfully made an application of the principle of typicality
to the argument of local realism versus quantum mechanics,
\emph{demonstrating how properly our framework based on the principle of typicality
works in practical problems in quantum mechanics}.
Thus, the results of this paper seem to \emph{further} confirm
our conjecture proposed in \cite{T18arXiv}
that \emph{the principle of typicality, Postulate~\ref{POT},
together with Postulates~\ref{state_space}, \ref{composition}, and \ref{evolution}
forms quantum mechanics}.

\section*{Acknowledgments}

This work was
supported by JSPS KAKENHI Grant Number~18K03405.

\end{document}